\documentclass[a4paper,UKenglish]{lipics-v2018}

\usepackage[ruled]{algorithm}
\usepackage{algpseudocode}

\usepackage{microtype}
\nolinenumbers 

\title{Towards a Theory of Mixing Graphs:
A Characterization of Perfect Mixability}

\titlerunning{Mixing Graphs}

\author{Miguel Coviello Gonzalez}{Department of Computer Science\\ University of California at Riverside}{}{}{}

\author{Marek Chrobak$\star$}{Department of Computer Science\\ University of California at Riverside}{}{}{}

\Copyright{The copyright is retained by the authors}

\authorrunning{M. Coviello Gonzalez and M. Chrobak}

\subjclass{%
}


\funding{Research supported by NSF grant CCF-1536026.}

\newcommand{\mareksmargincomment}[1]%
    {{%
      \marginpar{{\tiny\begin{minipage}{0.5in}
                       \begin{flushleft}
                          {\color{red}MCh} {#1}
                       \end{flushleft}
                       \end{minipage}
                }}
    }}

\newcommand{\etal}{{\emph{et~al.}}}

\newcommand{\braced}[1]{{ \left\{ #1 \right\} }}

\newcommand\ceil[1]{\lceil#1\rceil}

\newcommand{\assign}{\,{\leftarrow}\,}

\newcommand{\barp}{{\bar{p}}}

\newcommand{\barpi}{{\bar{\pi}}}

\newcommand{\calA}{{\cal A}}

\newcommand{\hatmu}{{\hat{\mu}}}

\newcommand{\tildee}{{\tilde{e}}}
\newcommand{\tildeE}{{\tilde{E}}}

\newcommand{\uC}{\check{C}}
\newcommand{\umu}{\check{\mu}}

\newcommand{\myparagraph}[1]{{\bigskip\noindent\textbf{#1}.}}
\newcommand{\mycase}[1]{{\underline{Case~#1}:}}

\newcommand{\NP}{{\mathbb{NP}}}

\newcommand{\ingroundset}{{\,\sqsubset\,}}

\newcommand{\integers}{{\mathbb Z}}
\newcommand{\posintegers}{{\mathbb{Z}_{> 0}}}
\newcommand{\nonnegintegers}{{\mathbb{Z}_{\ge 0}}}

\newcommand{\binrationals}{{\mathbb Q}_0}

\newcommand{\intermC}{{\uC}}
\newcommand{\intermmu}{{\umu}}

\newcommand{\nofdroplets}[1]{{|#1|}}
\newcommand{\nofconcentrations}[1]{{\|#1\|}}

\newcommand{\mysum}{{\textsf{sum}}}
\newcommand{\average}{{\textsf{ave}}}
\newcommand{\precision}{{\textsf{prec}}}

\newcommand{\mygcd}{{\textsf{gcd}}}

\newcommand{\diameter}{{\textsf{diam}}}

\newcommand{\potential}{{\Psi}}

\newcommand{\MixCond}{{\textrm{(MC)}}}

\newcommand{\MixProducibility}{{\textsc{MixProducibility}}}
\newcommand{\MixReachability}{{\textsc{MixReachability}}}
\newcommand{\PerfectMixability}{{\textsc{PerfectMixability}}}
\newcommand{\Partition}{{\textsc{Partition}}}

\newcommand{\LoC}{\text{LoC}}

\newcommand{\MinMix}{\texttt{Min-Mix}}
\newcommand{\DMRW}{\texttt{DMRW}}
\newcommand{\REMIA}{\texttt{REMIA}}

\newcommand{\WARA}{\texttt{WARA}}

\newcommand{\half}{{\textstyle\frac{1}{2}}}
\newcommand{\onehalf}{{\textstyle\frac{1}{2}}}

\newcommand{\onefourth}{{\textstyle\frac{1}{4}}}
\newcommand{\threefourths}{{\textstyle\frac{3}{4}}}

\newcommand{\oneeighth}{{\textstyle\frac{1}{8}}}

\newcommand{\seveneighths}{{\textstyle\frac{7}{8}}}

\newcommand{\onesixteenth}{{\textstyle\frac{1}{16}}}
\newcommand{\threesixteenths}{{\textstyle\frac{3}{16}}}
\newcommand{\fivesixteenths}{{\textstyle\frac{5}{16}}}
\newcommand{\sevensixteenths}{{\textstyle\frac{7}{16}}}
\newcommand{\ninesixteenths}{{\textstyle\frac{9}{16}}}

\newcommand{\seventhirtytwos}{{\textstyle\frac{7}{32}}}
\newcommand{\eleventhirtytwos}{{\textstyle\frac{11}{32}}}

\newtheorem{observation}{Observation}

\begin{document}
	
\maketitle

\begin{abstract}
	Some microfluidic lab-on-chip devices contain modules whose function is to mix two
	fluids, called reactant and buffer, in desired proportions. 
	In one of the technologies for fluid mixing the process
	can be represented by a directed acyclic graph whose nodes represent micro-mixers
	and edges represent micro-channels. 
	A micro-mixer has two input channels and two output channels;
	it receives two fluid droplets, one from each input, mixes them perfectly, and produces two
	droplets of the mixed fluid on its output channels. Such a mixing graph converts a
	set $I$ of input droplets into a set $T$ of output droplets, where the droplets are specified
	by their reactant concentrations. 
	The most fundamental algorithmic question related to mixing graphs is to determine,
	given an input set $I$ and a target set $T$, whether there is a mixing graph
	that converts $I$ into $T$. We refer to this decision problem as \emph{mix-reachability}.
	While the complexity of this problem remains open, we provide a 
	solution to its natural sub-problem, called \emph{perfect mixability}, 
	in which we ask whether, given a collection $C$ of droplets,
	there is a mixing graph that mixes $C$ perfectly,
	producing only droplets whose concentration is the average concentration of $C$.
	We provide a complete characterization of such perfectly mixable sets and 
	an efficient algorithm for testing perfect mixability. Further, we prove
	that any perfectly mixable set has a perfect-mixing graph of polynomial size, and that
	this graph can be computed in polynomial time.
\end{abstract}



\section{Introduction}
\label{sec: introduction}



\emph{Microfluidics} is an area of science and engineering dedicated to the study and 
manipulation of very small (picoliter to nanoliter~\cite{hsieh2012reagent}) amounts of 
fluids. Research advances in microfluidics led to the development of 
\emph{lab-on-chip} ($\LoC$) devices that integrate on a small chip
various functions of bulky and costly biochemical systems, including 
dispensing, mixing, and filtering of fluids, particle separation, and detection of chemicals.
$\LoC$s play increasingly important roles in applications that include
cancer research~\cite{li2013probing}, 
environment monitoring~\cite{marle2005microfluidic}, 
protein analysis~\cite{xu2010defect}, 
drug discovery~\cite{einav2008discovery} and physiological 
sample analysis~\cite{srinivasan2004integrated}.
The importance of $\LoC$ devices will soon scale up with the introduction of 
\emph{cloud laboratories}~\cite{hayden2014automated}, 
which give researchers access to state-of-the-art equipment and data analysis tools and
allow them to carry out their experiments remotely.

One of the most fundamental functions of $\LoC$ devices is \emph{mixing} of different fluids.
In particular, in applications related to sample preparation,
the objective is to produce desired volumes of pre-specified mixtures of fluids.
In typical applications only two fluids are involved, in which case
the process of mixing is often referred to as \emph{dilution}. The
fluid to be diluted is called \emph{reactant} and the diluting fluid is called \emph{buffer}.
For example, in clinical diagnostics common reactants include blood, serum, plasma and urine, while
phosphate buffered saline is often used as buffer~\cite{srinivasan2004integrated}.

There is a variety of different technologies that can be used to manufacture 
microfluidic devices for fluid mixing. In our work, we consider microfluidic
chips that involve a collection of tiny components called \emph{micro-mixers} connected
by \emph{micro-channels}. In such chips, input fluids are injected into the chip using fluid dispensers,
then they travel, following appropriate micro-channels, through a sequence of micro-mixers
in which they are subjected to mixing operations, and are eventually discharged into output reservoirs. 
We focus on \emph{droplet-based} chips, where fluids are manipulated in discrete units called \emph{droplets}. 
In such chips each micro-mixer has exactly two input and two output channels. 
It receives one droplet of fluid from
each input, mixes them perfectly, producing two identical droplets on its outputs.
Specifically, if the input droplets have (reactant) concentrations
$a,b$, then the produced droplets will have concentration $\half(a+b)$. 
It follows that all droplets flowing through the chip have concentrations
of the form $c/2^d$, where $c$ and $d\ge 0$ are integers. 
This simply means that their binary representations are finite, and we will
refer to such numbers simply as \emph{binary numbers}.
Throughout the paper we will assume (often tacitly) that
all concentration values are binary numbers. In this representation,
the number $d$, called \emph{precision}, is the number of
fractional bits (assuming $c$ is odd when $d\ge 1$). 

Processing of droplets on such chips can be naturally represented by a
directed acyclic graph $G$ that we call a \emph{mixing graph}. The edges of $G$
represent micro-channels. Source vertices (with in-degree $0$ and out-degree $1$) represent
dispensers, internal vertices (with in-degree and out-degree $2$)
represent micro-mixers, and sink vertices (with in-degree $1$ and out-degree $0$) represent
output reservoirs. Given a set $I$ of input droplets injected into the source nodes,
$G$ will convert it into a set $T$ of droplets in its sink nodes. We refer to this set $T$
as a \emph{target set}. 
An example of a mixing graph is shown in Figure~\ref{fig: mixing graph example}.
Here, and elsewhere in the paper, we represent each
droplet by its reactant concentration (which uniquely determines the 
buffer concentration, as both values add up to $1$).

	
\begin{figure}[ht]
	\begin{center}
		\includegraphics[width = 3.2in]{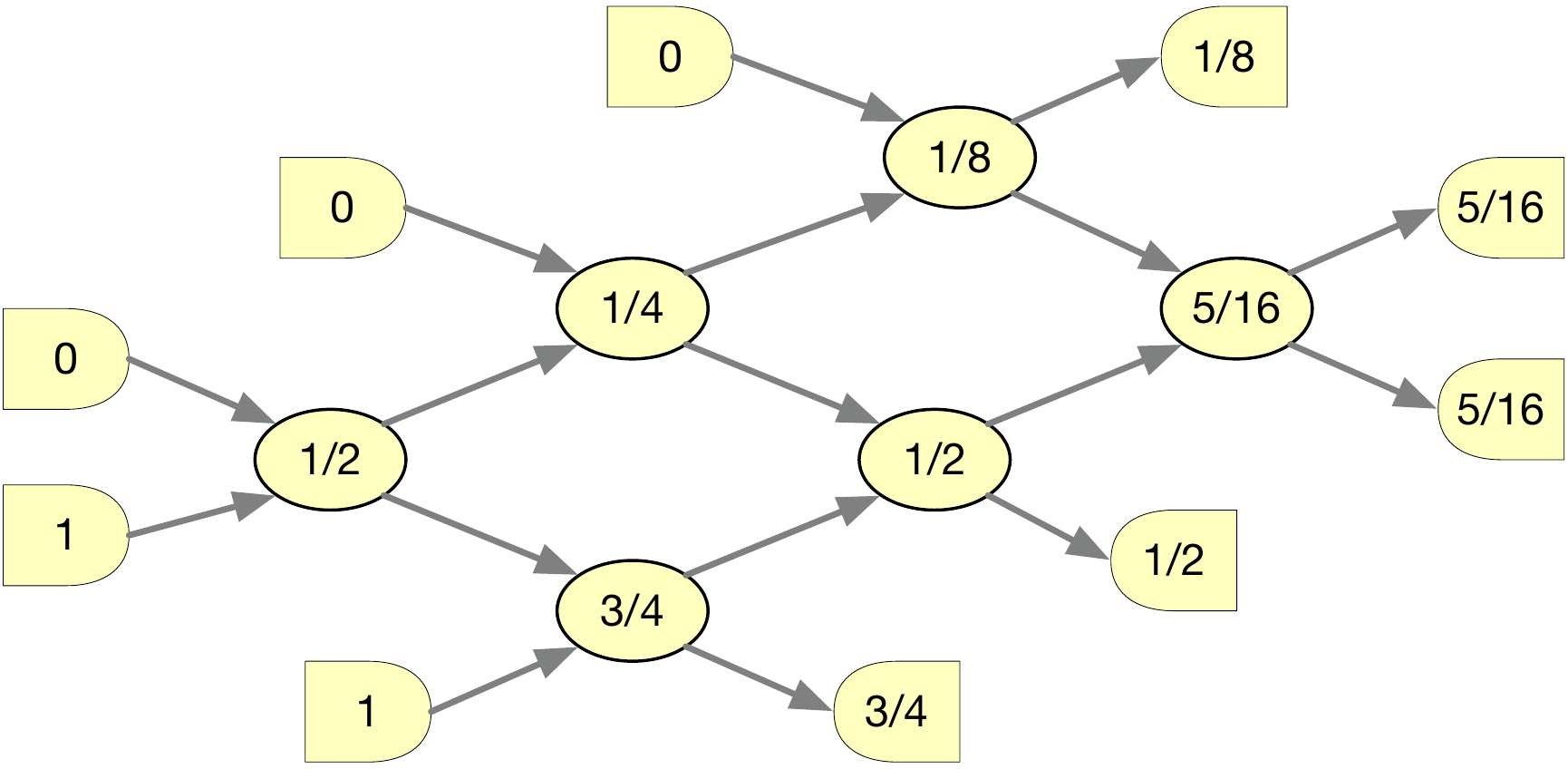}
		\caption{A mixing graph that produces 
		target set $T = \braced{\oneeighth,\fivesixteenths,\fivesixteenths,\onehalf,\threefourths}$
		from input set $I = \braced{0,0,0,1,1}$.
		Numbers on the micro-mixers (internal nodes) represent droplet concentrations produced
		by these mixers.}
		\label{fig: mixing graph example}
	\end{center}
\end{figure}

There is growing literature in the embedded systems and bioengineering communities
on designing microfluidic chips  represented by such mixing graphs.
The most fundamental algorithmic problem emerging in this
area is the following:

\begin{description}
	\item{{\MixReachability:}}
Given an input set $I$ and a target set $T$ of droplets with given reactant concentrations, 
design (if at all possible) a mixing graph that converts $I$ into $T$.
\end{description}

If there is a mixing graph that converts $I$ into $T$ then we say that 
$T$ is \emph{mix-reachable}, or just \emph{reachable}, from $I$.  
For $T$ to be reachable from $I$, clearly, $I$ and $T$ must have the
same cardinality and equal reactant volumes. However,  these conditions
are not sufficient. For example, $T = \braced{\onefourth,\threefourths}$
is not reachable from $I = \braced{0,1}$ (or from any other input set consisting only
of pure buffer and reactant droplets), because producing $\onefourth$ 
requires at least
two buffer droplets and one reactant droplet, but $T$ itself contains only two droplets.
 
In typical applications the input set $I$ consists of pure
reactant and buffer droplets (that is, only $0$'s and $1$'s). 
We denote this variant by {\MixProducibility}, and target sets reachable from such input
sets are caled \emph{mix-producible}, or just \emph{producible}.
{\MixProducibility} is not likely to be computationally easier than  {\MixReachability}. 
If $\calA$ is an algorithm that solves {\MixProducibility} then, via a simple linear
mapping, it can also solve the variant of {\MixReachability} where the input set
has droplets of \emph{any} two given concentrations (instead of $0$ and $1$). 
Further, consider the family $\braced{C_i}$ of concentration sets produced by up to 
some constant number $k$ of mixing operations from $I$. Then $\calA$ in fact determines whether $T$ is reachable
from at least one $C_i$, yet these sets $C_i$ have up to $k+2$ different
concentrations. While we do not have a formal reduction, these properties indicate that
taking advantage of
the special form of inputs in {\MixProducibility} may be challenging.


\myparagraph{Related work} 
The previous work in the literature focuses on designing mixing graphs that dilute the reactant
to produce some desired concentrations -- that is on the {\MixProducibility} problem.
To generate target sets that are not producible, one can consider mixing
graphs that besides a target set $T$ also produce some amount of superfluous fluid
called \emph{waste}. If we allow waste then, naturally, {\MixProducibility} can be extended 
to an optimization problem where the objective is to design a
mixing graph that generates $T$ while minimizing waste. Alternative
objective functions have been studied, for example minimizing the
reactant waste, minimizing the number of micro-mixers, and other.

Most of the previous papers on this topic study designing mixing graphs using
heuristic approaches. Earlier studies
focused on producing \emph{single-concentration targets}, where only one droplet of some 
desired concentration is needed. 
This line of research was pioneered by Thies~{\etal}~\cite{thies2008abstraction},
who proposed an algorithm  called $\MinMix$ 
that constructs a mixing graph for a single target droplet.
Roy~{\etal}~\cite{roy2010optimization} developed a single-droplet
algorithm called $\DMRW$ that considered waste reduction and the number of mixing operations.
Huang~{\etal}~\cite{huang2012reactant} and Chiang~{\etal}~\cite{chiang2013graph}
proposed single-droplet algorithms designed to minimize reactant usage.

Many applications, however, require target sets with multiple concentrations
(see, for example, 
\cite{xu2008automated,xu2010defect,srinivasan2004droplet,srinivasan2004integrated,hsieh1998automated}).
Target sets that arise in sample preparation typically involve concentration values that form
arithmetic or geometric sequences
(referred to, respectively, as ``linear'' and ``logarithmic'' in some literature -- see, for example~\cite{kim2008serial}), 
but the special form of such sets does not seem to facilitate the design of mixing graphs.
For multiple-concentration targets, Huang~{\etal}~\cite{huang2013reactant} proposed an 
algorithm called $\WARA$, which is an extension of Algorithm~$\REMIA$ from~\cite{huang2012reactant}.
Mitra~{\etal}~\cite{mitra2012chip} model the problem of producing multiple concentrations
as an instance of the Asymmetric TSP on a de Brujin graph.

The papers cited above describe heuristic algorithms with no formal performance
guarantees. Dinh~{\etal}~\cite{dinh2014network} take a more rigorous approach.
They model the problem as an equal-split flow problem \cite{meyers2009integer} on a
``universal'' graph that contains all possible mixing graphs of depth at most $d$ as subgraphs,
where $d$ is the maximum precision in $T$. 
By assigning appropriate capacity and cost values to edges, the problem of 
extracting a mixing subgraph that minimizes waste
can be represented as an integer linear program, resulting
in an algorithm that is doubly exponential in $d$. Unfortunately, contrary to the claim in~\cite{dinh2014network},
their algorithm does not necessarily produce mixing graphs with minimum waste. The
reason is, as we show in Appendix~\ref{sec: counter-example for algorihtm of Dinh},
that there are target sets with maximum precision $d$ that require mixing graphs of depth
larger than $d$ to be produced without waste.


\myparagraph{Our results}
To our knowledge,  the computational complexity of {\MixReachability} is open; in fact,
(given the flaw in~\cite{dinh2014network} mentioned above) it is not even known
whether the {\MixProducibility} variant is decidable. This paper
reports partial progress towards resolving this problem.
We consider the following sub-problem of {\MixReachability}: 

\begin{description}
	\item{{\PerfectMixability:}}
	Given a set $C$ of $n$ droplets with binary concentrations
	and binary average value $\mu = (\sum_{c\in C}c)/n$,
	is there a mixing graph that mixes $C$ perfectly, converting $C$ into
	the set of $n$ droplets of concentration $\mu$?
\end{description}

	
\begin{figure}[ht]
	\begin{center}
		\includegraphics[width = 3.7in]{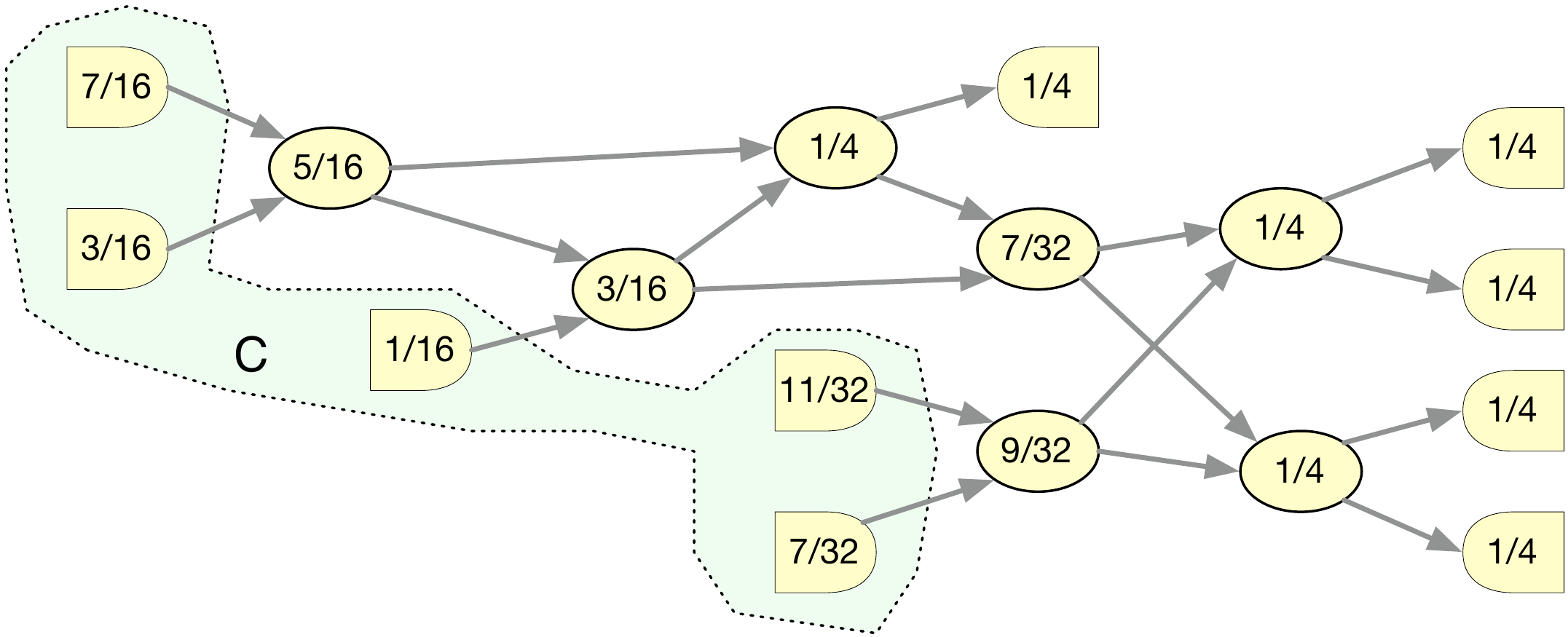}
		\caption{A mixing graph that perfectly mixes
		set $C = \braced{\onesixteenth, \threesixteenths,\seventhirtytwos,\eleventhirtytwos,\sevensixteenths}$.}
		\label{fig: dilution example}
	\end{center}
\end{figure}

Figure~\ref{fig: dilution example} shows an example of a perfect-mixing graph.
As an example of a set that is not perfectly mixable,
consider $D = \braced{0,\threesixteenths,\ninesixteenths}$. After any (non-zero) number of
mixing operations the resulting set of concentrations
will have the form $D' = \braced{a,a,b}$ for $a\neq b$, so no finite mixing graph will
convert $D$ into its perfect mixture $\braced{\onefourth,\onefourth,\onefourth}$.

In this paper, addressing the {\PerfectMixability} problem, 
we provide a complete characterization of perfectly mixable
droplet sets (with binary concentration values), and show that there is a polynomial-time 
algorithm that tests whether a given set is perfectly mixable,
and if so, constructs a polynomial-size perfect-mixing graph for it.

We represent droplet sets as multisets of concentration values.
First, we observe that without loss of generality we can assume that 
$C\cup\braced{\mu}\ingroundset \integers$, for otherwise we can simply 
rescale all values by an appropriate power of $2$. 
($\integers$ is the set of integers; $\posintegers$ and $\nonnegintegers$ are the sets of positive
and non-negative integers, respectively. Symbol $\ingroundset$ is used to specify a ground set of a multiset.)
For any finite multiset $A\ingroundset\integers$ and $b\in\posintegers$, we define
$A$ to be \emph{$b$-congruent} if $x\equiv y \pmod{b}$ for all $x,y\in A$.
(Otherwise we say that $A$ is \emph{$b$-incongruent}.)

We say that $C$ satisfies \emph{Condition~$\MixCond$} if, 
for each odd $b\in\posintegers$, if 
$C$ is $b$-congruent then $C\cup \braced{\mu}$ is $b$-congruent as well, 
where $\mu = \average(C)$.
The following theorem summarizes our results.


\begin{theorem}\label{thm: miscibility characterization}
Assume that $n\ge 4$ and $C\cup \braced{\mu}\ingroundset\integers$, where 
$\mu = \average(C)$. Then:
\begin{description}
	\item[(a)] $C$ is perfectly mixable if and only if $C$ satisfies Condition~$\MixCond$.
	\item[(b)] If $C$ satisfies Condition~{\MixCond} then it can be perfectly mixed
	with precision at most $1$ and in a polynomial number of steps. 
	(That is, $C$ has a perfect-mixing graph of polynomial size where 
	all intermediate concentration values are half-integral.)
	\item[(c)] There is a polynomial-time algorithm that tests whether $C$ 
	is perfectly mixable and, if so, computes a polynomial-size perfect-mixing graph for $C$.
\end{description}
\end{theorem}


Part~(b) implies that, in general (if the concentrations in $C\cup\braced{\mu}$ are arbitrary
binary values), if $C$ is perfectly mixable at all then
$C$ can be mixed perfectly with precision at most
$d+1$, where $d$ is the maximum precision in $C\cup\braced{\mu}$; 
in other words, at most one extra bit of precision is needed
in the intermediate nodes of a perfect-mixing graph for $C$.

This extra 1-bit of precision in part~(b) of Theorem~\ref{thm: miscibility characterization}
is necessary. For example, $C = \braced{0,0,0,3,7}$ (for which $\mu = 2$) cannot be
mixed with precision $0$. If we mix $3$ and $7$, we will obtain multiset
$\braced{0,0,0,5,5}$ which is not perfectly mixable, as it violates Condition~{\MixCond}.
Any other mixing creates fractional values.
However, $C$ does have a mixing graph where the intermediate precision is at most $1$
--- see Figure~\ref{fig: precision 1 example}.

	
\begin{figure}[ht]
	\begin{center}
		\includegraphics[width = 3.7in]{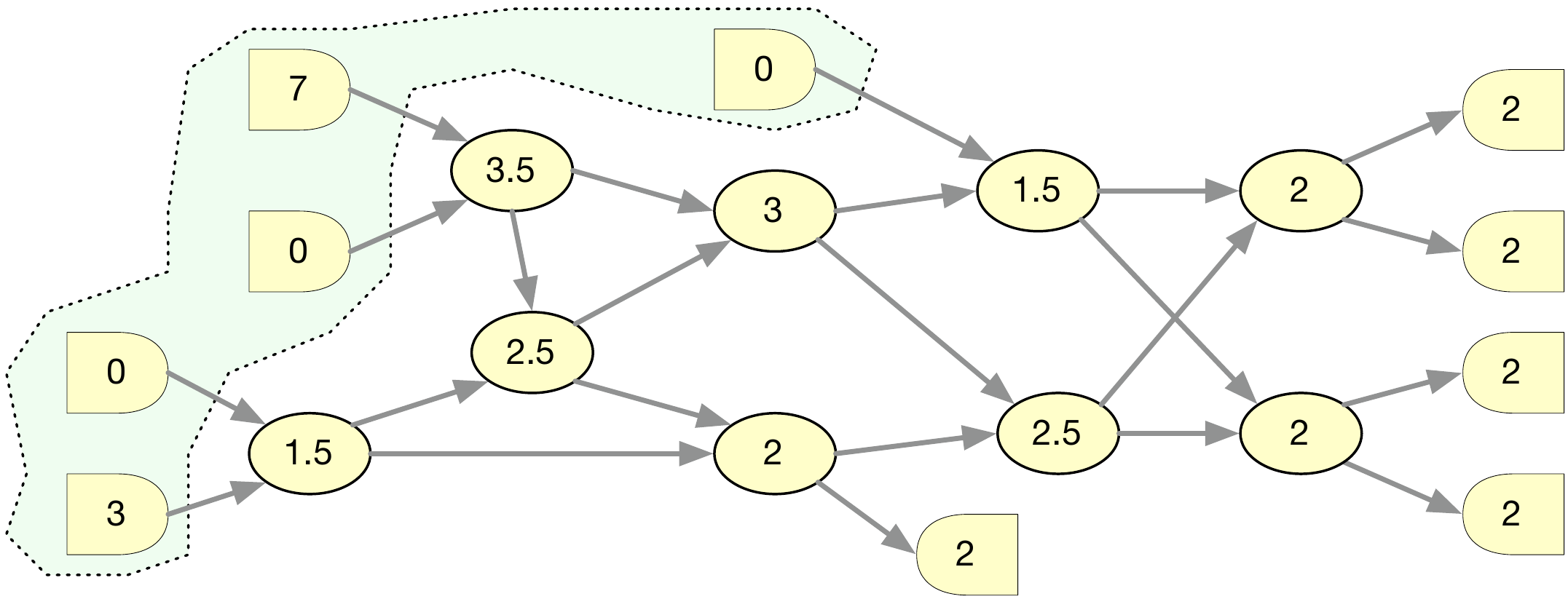}
		\caption{A perfect mixing graph for $C = \braced{0,0,0,3,7}$ with precision $1$.}
		\label{fig: precision 1 example}
	\end{center}
\end{figure}

Note that in Theorem~\ref{thm: miscibility characterization} we assume that $n = |C|\ge 4$.
Regarding smaller values of $n$, for $n\leq 2$, trivially, all configurations $C$ with at most two droplets
are perfectly mixable with precision $0$.
The case $n=3$ is exceptional, as in this case
Theorem~\ref{thm: miscibility characterization} is actually false. (For example,
consider configuration $C = \braced{0,1,5}$, for which $\mu = 2$. This configuration
is $b$-incongruent for all odd $b>1$, so it satisfies condition~{\MixCond}, but is
not perfectly mixable.) 
Nevertheless, for $n=3$, perfectly mixable configurations are easy to
characterize: Let $C = \braced{a,b,c}$, where $a\le b \le c$.
Then $C$ is perfectly mixable if and only if $b = \half (a+c)$.
Further, if this condition holds, $C$ is perfectly mixable with precision $0$.
(That this condition is sufficient is obvious. That it is also necessary can
be proven by following the argument 
for the example configuration $D$ right after the definition of {\PerfectMixability}.)

To clarify the polynomial bounds in Theorem~\ref{thm: miscibility characterization},
we assume that the input configuration $C$ is specified by listing the concentration
values of individual droplets, and its size is the total number of bits
used in this representation. For the sake of concretness, we will
assume that $C$ is already rescaled to consist only of integers,
and we will define the input size as $s(C) = \sum_{c\in C}\log(|c|+2)$.
This value is within a small constant factor of the actual number of
bits representing $C$. 
Consequently, the polynomial bounds in Theorem~\ref{thm: miscibility characterization} 
are with respect to $s(C)$.


\myparagraph{Overview of the paper}
The proof of Theorem~\ref{thm: miscibility characterization}
is given in several sections. The necessity of Condition~{\MixCond} in 
Theorem~\ref{thm: miscibility characterization}(a)
is relatively simple to show; the proof appears in Section~\ref{sec: necessity of mc}. 
The proof that Condition~{\MixCond} is sufficient is more challenging. We first
show in Section~\ref{sec: some auxiliary lemmas} (see Corollary~\ref{cor: mc prime power factors})
that, in essence, in Condition~{\MixCond} it is sufficient to consider only
the values of $b$ that do not exceed the maximum concentration in $C$ and
are powers of prime factors of $n$.
This property is used in Section~\ref{sec: sufficiency of condition mc}
to show that any set $C$ that satisfies Condition~{\MixCond} has 
a perfect-mixing graph, completing the proof of Theorem~\ref{thm: miscibility characterization}(a).
The mixing graph constructed in Section~\ref{sec: sufficiency of condition mc} 
has precision at most $1$, proving the frst part of Theorem~\ref{thm: miscibility characterization}(b).
The second part, showing the existence of a perfect-mixing graph of polynomial
size is established in Section~\ref{sec: polynomial proof}.
The proof of Theorem~\ref{thm: miscibility characterization}(c) is divided into
two parts: that testing Condition~{\MixCond} can be done in polynomial
time follows directly from Corollary~\ref{cor: mc prime power factors} 
in Section~\ref{sec: some auxiliary lemmas}, while a polynomial-time 
algorithm for constructing a perfect-mixing graph is described
in Sections~\ref{sec: polynomial proof} and~\ref{sec: polynomial running time}.


\section{Preliminaries}
\label{sec: preliminaries}



Let $\binrationals$ be the set of  binary numbers.
For $c\in\binrationals$, we denote by $\precision(c)$ the precision of $c$, that is the
number of fractional bits in the binary representation of $c$, assuming there are no trailing $0$'s.
In other words, $\precision(c)$ is the smallest $d\in\nonnegintegers$ such that $c = a/2^d$ for some $a\in\integers$. 
If $c = a/2^d$ represents actual fluid concentration, then we have $0\le a \le 2^d$. However,  
it is convenient to relax this restriction and allow ``concentration values'' that are
arbitrary binary numbers, even negative. In fact, as we show shortly, 
it will be convenient to work with integral values.

By a \emph{configuration} we mean a multiset of $n$ binary numbers, called droplets or concentrations.
In the literature, multisets are often represented by their characteristic functions (that specify
the multiplicity of each element of the ground set). In this paper we will generally use set-theoretic
terminology, with its natural interpretation. For example, for a configuration $C$ and a concentration $a$, 
$a\in C$ means that the multiplicity of $a$ in $C$ is strictly positive, while $a\notin C$ means that it's zero.
Or, $C - \braced{a} = C'$ means that the multiplicity of $a$ in $C'$ is one less than in $C$, while
other multiplicities are the same.
The number of droplets in $C$ is denoted $\nofdroplets{C}=n$, 
while the number of different concentrations is denoted $\nofconcentrations{C}=m$.
We will typically denote a configuration by
$C = \braced{f_1:c_1,f_2: c_2,...,f_m:c_m} \ingroundset \binrationals$, where each $c_i$ represents
a (different) concentration value and $f_i$ denotes the
multiplicity of $c_i$ in $C$, so that $\sum_{i=1}^m f_i = n$. 
Occasionally, if it does not lead to confusion, we may say
``droplet $c_i$'' or ``concentration $c_i$'', referring to some droplet with concentration $c_i$.
If $f_i=1$, we shorten ``$f_i:c_i$'' to just ``$c_i$''.
If $f_i = 1$ we say that droplet $c_i$ is a \emph{singleton}, 
if $f_i = 2$ we say that droplet $c_i$ is a \emph{doubleton}
and if $f_i\geq 2$ we say that droplet $c_i$ is a \emph{non-singleton}.
By $\mysum(C)$ we denote the sum of $C$, that is $\mysum(C) = \sum_{c\in C} c$.
$\average(C) = \mysum(C)/n$ is the average value of the concentrations in $C$ and will
be typically denoted by $\mu$.
(Later, we will typically deal with configurations $C$ such that
$C\cup\braced{\mu}\ingroundset\integers$.)

Mixing graphs were defined in the introduction. As we are not concerned in this paper with the
topological properties of mixing graphs, we will often identify a mixing graph $G$ with a
corresponding \emph{mixing sequence}, which is a sequence (not necessarily unique) of mixing operations
that convert $C$ into its perfect mixture. In other words, a mixing sequence is a sequence
of mixing operations in a topological ordering of a mixing graph.

Of course in a perfect-mixing graph (or sequence) $G$ for $C$, all concentrations in $G$,
including those in $C\cup\braced{\mu}$,
must have finite precision (that is, belong to $\binrationals$);
in fact, the maximum concentration in $G$ is at least $\max\braced{\precision(C),\precision(\mu)}$.
In addition to the basic question about finding a perfect-mixing graph for $C$, 
we are also interested in bounding the precision required to do so.

For $x\in\binrationals$, define multisets $C+x = \braced{c+x \mid c\in C}$,
$C-x = C+(-x)$, and $C\cdot x = \braced{c\cdot x \mid c\in C}$.
The next observation says that offsetting all values in $C$
does not affect perfect mixability, as long as the offset value's precision
does not exceed that of $C$ or $\mu$.

\begin{observation}
\label{obs: offsetting does not affect reachability}
Let $\mu = \average(C)$ and $x\in\binrationals$. 
Also, let $d\in\posintegers$ be such that $d\ge \max\braced{\precision(C),\precision(\mu),\precision(x)}$.
Then $C$ is perfectly mixable with precision $d$ if and only if
$C' = C+x$ is perfectly mixable with precision $d$.
\end{observation}

\begin{proof}
$(\Rightarrow)$
Suppose that $G$ is a perfect-mixing sequence for $C$ with precision $d$.
Run the same sequence $G$ on input $C'$. If some mixing step in $G$ produces a value $z$ when
the input is $C$ then on $C'$ its value is $z+x$,
and $\precision(z+x)\le \max\braced{\precision(z),\precision(x)} \le d$.
Thus the maximum precision in $G$ for input $C'$ is at most $d$. 

$(\Leftarrow)$
The proof for this implication follows from noting that
$\mu' = \average(C') = \mu+x$,
$\max\braced{\precision(C'),\precision(\mu')}\le d$, and by applying the
above argument to $-x$ instead of $x$.
\end{proof}

\begin{observation}
\label{obs: power of 2 rescaling}
Let $\mu = \average(C)$, $\delta = \max\braced{\precision(C),\precision(\mu)}$, 
$C' = C\cdot 2^\delta$ with $\mu'=\average(C')= 2^\delta\mu$. 
	(Thus $C'\cup\braced{\mu'}\ingroundset\integers$.)
	Then $C$ is perfectly mixable with precision $d\ge\delta$
	if and only if $C'$ is perfectly mixable with precision $d' = d-\delta$. 
\end{observation}

\begin{proof}
$(\Rightarrow)$
Let $G$ be a perfect-mixing sequence for $C$, with precision $d$.
Run the same sequence $G$ on input $C'$.
If some node in $G$ produces a value $z$ on input $C$, then its value on input $C'$ will be
$z 2^\delta$, and $\precision(z 2^\delta) = \max\braced{\precision(z)-\delta,0} \le d-\delta = d'$.

$(\Leftarrow)$
Let $G'$ be a perfect-mixing sequence for $C'$ with precision $d'$.
Run $G'$ on input $C$. If some node in $G'$ produces a value $y$ on input $C'$, then
its value on input $C$ will be $y/2^\delta$, and
$\precision(y/2^\delta) \leq \precision(y)+\delta \le d'+\delta = d$.
\end{proof}


\myparagraph{Integral configurations}
Per Observation~\ref{obs: power of 2 rescaling}, we can restrict our attention to configurations with integer
values and average, that is, we will be assuming that $C\cup\braced{\mu}\ingroundset \integers$.

For $x\in\posintegers$, if each $c\in C$ is a multiple of $x$,
let $C/x = \braced{c/x \mid c\in C}$.
For integral configurations, we can extend Observation~\ref{obs: power of 2 rescaling}
to also multiplying $C$ by an odd integer or dividing it by a common odd factor of
all concentrations in $C$. 

\begin{observation}
\label{obs: odd integer rescaling}
Assume that $C\cup\braced{\mu}\ingroundset \integers$ and let $x\in\posintegers$ be odd.
\begin{description}
\item{\emph{(a)}} Let $C' = C\cdot x$. Then
	$C$ is perfectly mixable with precision $0$ if and only if
	$C'$ is perfectly mixable with precision $0$.	
\item{\emph{(b)}} Suppose that $x$ 
is a divisor of all concentrations in $C\cup\braced{\mu}$.
Then $C$ is perfectly mixable with precision $0$ if and only if
$C/x$ is perfectly mixable with precision $0$.
\end{description}	
\end{observation}

\begin{proof}
Part~(b) follows from~(a), so we only prove part~(a).
Any sequence $G$ of mixing operations for $C$ can be applied to $C\cdot x$. 
By simple induction,
if some intermediate value in $G$ was an integer $z$, now its value will be $z x$,
also an integer. This shows the $(\Rightarrow)$ implication.

To justify the $(\Leftarrow)$ implication, 
suppose that $G'$ is a sequence of mixing operations for $C'$ and that all
concentrations in $G'$ are integer. 
Since $x$ is odd and all concentrations in $C'$ are multiples of $x$, 
every concentration in $G'$ will be also a multiple of $x$, including  $\average(C')$.  
Thus, if  we run $G'$ on $C$ instead of $C'$, if some node's concentration
was $c x$, now it will be $c$.
Thus, the $(\Leftarrow)$ implication holds.
\end{proof}



\section{Necessity of Condition~(MC)}
\label{sec: necessity of mc}



In this section we prove that Condition~{\MixCond} in 
Theorem~\ref{thm: miscibility characterization}(a) is necessary for perfect mixability.
So let $C\cup \braced{\mu}\ingroundset\integers$, where $\mu = \average(C)$, and assume that
$C$ is perfectly mixable. Let $G$ be a graph (or a sequence) that mixes $C$ perfectly.
We want to prove that $C$ satisfies Condition~{\MixCond}.

Suppose that $C$ is $b$-congruent for some odd $b\in\posintegers$.
Consider an auxiliary configuration
$C' = C\cdot 2^\delta$, where $\delta$ is sufficiently large, so that all
intermediate concentrations in $G$ when applying $G$ to $C'$ are integral. 
This $C'$ is $b$-congruent, and starting from $C'$, $G$ produces a perfect mixture of $C'$,
that is $\braced{n:\mu'}$, for $\mu' = 2^\delta\mu$.

Since $C'$ is $b$-congruent, there is $\beta\in\braced{0,...,b-1}$ such that
for each $x\in C'$ we have $x\equiv \beta \pmod{b}$.
We claim that this property is preserved as we apply mixing operations to
droplets in $C'$. Indeed, suppose that we mix two droplets with concentrations
$x,y\in C'$, producing two droplets with concentration $z$.  
Since $x\equiv \beta \pmod{b}$ and $y\equiv \beta \pmod{b}$,
we have $x = \alpha b + \beta$ and $y = \alpha' b + \beta$, for some $\alpha, \alpha'\in \integers$,
so $z = \half(x+y) = (\half(\alpha+\alpha'))b + \beta$. As $b$ is odd
(and $z$ is integer), $\alpha+\alpha'$ must be even, and therefore
$z\equiv \beta \pmod{b}$, as claimed. Eventually $G$ produces $\mu'$,
so this must also hold for $z = \mu'$. This implies that 
$C'\cup\braced{\mu'}$ is $b$-congruent.

Finally, since $C'\cup\braced{\mu'}$ is $b$-congruent, for
all $2^\delta x, 2^\delta y\in C'\cup\braced{\mu'}$ it holds
that $2^\delta x \equiv 2^\delta y \pmod{b}$. 
But this implies that  $x \equiv y \pmod{b}$, because
$b$ is odd. So we can conclude that 
$C\cup\braced{\mu}$ is $b$-congruent, thus proving that
$C$ satisfies Condition~{\MixCond}.


\section{Some Auxiliary Lemmas}
\label{sec: some auxiliary lemmas}




In this section we show that, for Condition~{\MixCond},
we only need to consider $b$'s that are odd-prime factors of $n$
and satisfy $b\leq c_{max}$,
where $c_{max}$ denotes the maximum absolute value of concentrations in $C$.
These properties play an important role in the sufficiency proof of Condition~{\MixCond} in 
Theorem~\ref{thm: miscibility characterization}(a).
Additionally, they lead to an efficient algorithm for testing perfect mixability
(see Section~\ref{sec: polynomial running time}), 
thus showing the first part of Theorem~\ref{thm: miscibility characterization}(c). 


\begin{lemma}\label{lem: uniform wrt coprime numbers}
Let $b,c\in\posintegers$ and $A\ingroundset \integers$.
(a) If $A$ is $bc$-congruent then $A$ is also $b$-congruent.
(b) If $\mygcd(b,c) = 1$ and $A$ is both $b$-congruent and $c$-congruent
	then $A$ is $bc$-congruent.
\end{lemma}

\begin{proof}
Part~(a) is trivial, because $x\equiv y\pmod{bc}$ implies that
$x\equiv y\pmod{b}$.
Part~(b) is also simple: Suppose that $x\equiv y\pmod{b}$
and $x\equiv y\pmod{c}$. This means that $b|(x-y)$ and $c|(x-y)$.
This, since $b,c$ are co-prime, implies that $(bc)|(x-y)$, 
which is equivalent to $x\equiv y\pmod{bc}$.
\end{proof}


\begin{lemma}\label{lem: mc powers of primes}
If Condition~{\MixCond} holds for all $b\in\posintegers$ that are a power of an odd prime
then it holds for all odd $b\in\posintegers$.
\end{lemma}

\begin{proof}
Assume that condition~{\MixCond} holds for all $b$ that are odd prime powers. 
Let $b'\in\posintegers$ be odd, with factorization
$b' = p_1^{\gamma_1} ... p_k^{\gamma_k}$,
for different odd primes $p_1 , ...,p_k$, and suppose that $C$ is $b'$-congruent.
Then, by Lemma~\ref{lem: uniform wrt coprime numbers}(a),
$C$ is also  $p_i^{\gamma_i}$-congruent for all $i$.
Since condition~{\MixCond} holds for $p_i^{\gamma_i}$,
this implies that $C\cup\braced{\mu}$ is $p_i^{\gamma_i}$-congruent for all $i$.
By repeated application of Lemma~\ref{lem: uniform wrt coprime numbers}(b), we then obtain that
$C\cup\braced{\mu}$ is $b'$-congruent as well. So Condition~{\MixCond} holds for $b'$.
\end{proof}


\begin{lemma}\label{lem: b's co-prime to n}
Let $b\in\posintegers$ be odd. If $\gcd(b,n) = 1$ then
Condition~{\MixCond} holds for $b$.
\end{lemma}

\begin{proof}
Assume that $C$ is $b$-congruent. 
By Observation~\ref{obs: offsetting does not affect reachability}, without loss of generality
we can assume that all numbers in $C$ are multiples of $b$. 
(Otherwise we can consider $C' = C-c$, for an arbitrary $c\in C$, 
because Condition~{\MixCond} is not affected by offsetting $C$.)
Thus $\mysum(C) = b\beta$, for some $\beta\in\integers$,
which gives us that $\mu = \mysum(C)/n = b\beta/n$.
As $\mu$ is integer and $\gcd(b,n) = 1$, $\beta$ must be a multiple of $n$.
We can thus conclude that $\mu$ is a multiple of $b$.
This means that $C\cup \braced{\mu}$ is $b$-congruent, proving that
Condition~{\MixCond} holds for $b$.
\end{proof}


\begin{corollary}\label{cor: mc prime power factors}
	Let $c_{max}$ be the maximum absolute value of concentrations in $C$. 
	If Condition~{\MixCond} holds for all $b\in\posintegers$ with $b\leq c_{max}$
	that are powers of odd prime factors of $n$,
	then Condition~{\MixCond} holds for all odd $b$.
\end{corollary}


To substantiate Corollary~\ref{cor: mc prime power factors},
note that $C$ satisfies Condition~$\MixCond$ 
for all odd $b\in\posintegers$ larger than $c_{max}$.
This holds because for each such $b$, $c\mod b=c$ for all $c\in C$, 
so the remainders of the (different) concentrations in $C$ modulo $b$ are all different. 
(In the trivial case where 
all concentrations in $C$ are equal, $\mu$ is also equal and thus
$C$ satisfies Condition~{\MixCond}; such $C$ is actually perfectly mixed already.)


\section{Sufficiency of Condition~(MC)}
\label{sec: sufficiency of condition mc}



In this section we prove that Condition~{\MixCond} in Theorem~\ref{thm: miscibility characterization}(a) 
is sufficient for perfect mixability. A perfect-mixing graph constructed in our argument has
precision at most $1$, showing also the first part of Theorem~\ref{thm: miscibility characterization}(b).

Assume that $n\ge 4$. Let $C$ with $\mu = \average(C)$ and 
$C\cup\braced{\mu}\ingroundset\integers$ be the input configuration,
and assume that $C$ satisfies Condition~{\MixCond}. 
The outline of our proof is as follows:
\begin{itemize}
	\item First we prove that $C$ is perfectly mixable with precision $0$ when $n$ is a power of $2$.
		This easily extends to configurations $C$ called \emph{near-final}, 
		which are disjoint unions of multisets with the same average and cardinalities being powers of $2$. 
		In particular, this proves Theorem~\ref{thm: miscibility characterization}(a) for $n=4$.
		
	\item Next, we give a proof for $n\ge 7$. The basic idea of the proof is
		to define an invariant~(I) and show that any configuration that
		satisfies~(I) has a pair of droplets whose mixing either preserves
		invariant~(I) or produces a near-final configuration. 
		Condition~(I) is stronger than~{\MixCond} (it implies~{\MixCond}, but not
		vice versa), but we show that
		that any configuration that satisfies Condition~{\MixCond} can be
		modified to satisfy~(I).
		
	\item We then give separate proofs for $n=5,6$.
		The proof for $n=5$ is similar to the case $n\ge 7$, but
		it requires a more subtle invariant.
		The proof for $n=6$ is derived by minor modifications to the proof for $n=5$.
\end{itemize}


\subsection{Perfect Mixability of Near-Final Configurations}
\label{subsec: near-final configurations}



Let $C\ingroundset\integers$ be a configuration with
$\nofdroplets{C} = n = \sigma 2^\tau$, for some odd $\sigma\in\posintegers$ and $\tau\in\nonnegintegers$,
with $\average(C) = \mu\in\integers$.
We say that $C$ is \emph{near-final} if it can be partitioned
into disjoint multisets $C_1,C_2,...,C_k$, such that, for each $j$,
$\average(C_j) = \mu$ and $\nofdroplets{C_j}$ is a power of $2$.
In this sub-section we show (Lemma~\ref{lem: mixability for near-final} below)
that near-final configurations are perfectly mixable with precision $0$.
We also show that configurations with
only two different concentrations that satisfy Condition~{\MixCond} are near-final,
and thus perfectly mixable.

Define $\potential(C) = \sum_{c\in C}(c - \mu)^2$, which can be thought of as the
un-normalized variance of $C$. Obviously $\potential(C) \in \nonnegintegers$,
$\potential(C)=0$ if and only if $C$ is a perfect mixture,
and, by a straightforward calculation, mixing any two different same-parity concentrations in $C$ 
decreases the value of $\potential(C)$ by at least $1$.


\begin{lemma}\label{lem: mixability for near-final} 
	If $C$ is near-final
	then $C$ is perfectly mixable with precision $0$.
\end{lemma}

\begin{proof}
It is sufficient to prove the lemma for the case when $n$ is a power
of $2$. (Otherwise, we can apply it separately to each set $C_j$ in the 
partition of $C$ from the definition of near-final configurations.)	

So assume that $n$ is a power of $2$.
It is sufficient to show that if $\nofconcentrations{C} = m \neq 1$ (that is,
$C$ is not yet perfectly mixed) then $C$ contains
	two different concentrations with the same parity. 
	(Each such mixing strictly decreases $\potential(C)$, so a finite sequence
	of such mixing operations will perfectly mix $C$.)
	This is trivially true when $m\geq 3$,
	so it is sufficient to prove it for $m=2$,
	that is for $C = \braced{f_1:c_1, f_2:c_2}$. Without loss of generality, 
	by Observation~\ref{obs: offsetting does not affect reachability},
	we can assume that $c_2 = 0$,
	and then we claim that $c_1$ is even.
	We have $\average(C) = \mu = f_1 c_1/n$.
	As $\mu\in\integers$, $f_1 < n$ and $n$ is a power of $2$,
	we have that $c_1$ must be even, as claimed.
\end{proof}


\begin{lemma}\label{lem: mixability m = 2}
	Assume that $\nofconcentrations{C}=2$, say $C=\braced{f_1:c_1,f_2:c_2}$,
	and that $C$ satisfies Condition~{\MixCond}.
	Then $\sigma$ divides $f_1$ and $f_2$.
	Consequently, we have that
	$n$ is not prime and $C$ is near-final.
\end{lemma}

\begin{proof}
Without loss of generality, assume that $c_2 = 0$. (Otherwise consider
$C' = C-c_2$ instead. This does not affect Condition~~{\MixCond} and the property of being near-final.)
Let $c_1 = \alpha 2^{\gamma}$, for some odd $\alpha\in\posintegers$. 
Then $\mu = f_1c_1/n = f_1\alpha 2^{\gamma}/(\sigma 2^\tau)$. 
Since $\alpha$ divides $c_1$ and $c_2$, Condition~{\MixCond} implies that 
$\alpha$ must also divide $\mu$. In other words, $\mu/\alpha = f_1 2^{\gamma}/(\sigma 2^\tau)$
is integer. This implies, in turn, that $f_1$ is a multiple of $\sigma$, as claimed.
Since $f_2 = n-f_1$, it also gives us that $f_2$ is a multiple of $\sigma$.    

This immediately implies that $n$ cannot be prime, for $n=f_1+f_2$ is a sum of
two non-zero multiples of $\sigma$.

To prove the last claim,
let $f_1 =\sigma f'_1$ and $f_2 = \sigma f'_2$,  for some $f'_1, f'_2 \in\posintegers$.
Partition $C$ into $\sigma$ sub-multisets of the form $C_j = \braced{f_1':c_1,f'_2:0}$, for $j=1,2,...,\sigma$.
The cardinality of each $C_j$ is $f'_1+f'_2 = n/\sigma = 2^\tau$ and
its average is $\average(C_j) = f'_1c_1/(f'_1+f'_2) = (f_1/\sigma)c_1/(n/\sigma) = f_1c_1/n = \mu$.
Therefore $C$ is near-final, as claimed.
\end{proof}

We remark that when $n$ is a power of $2$ there is an alternative way to
perfectly mix $C$, using a divide-and-conquer approach: partition $C$ into
two equal-size multisets $C', C''$, mix each of these recursively, obtaining
$n/2$ droplets with concentration $\mu' = \average(C')$ and $n/2$ droplets
with concentration $\mu'' = \average(C'')$, 
and then mix $n/2$ disjoint pairs of droplets, one $\mu'$ and the other $\mu''$,
producing $n$ droplets with concentration $\mu = \average(C)$.
This approach, however, produces mixing graphs where the intermediate precision 
could be quite large, so it is not sufficient for our purpose.


\subsection{Proof for arbitrary $n \geq 7$}
\label{subsec: proof for n greater-equal than 7}



In this sub-section we prove that Condition~{\MixCond} in Theorem~\ref{thm: miscibility characterization}(a)
is sufficient for perfect mixability when $n\geq 7$. 
Let $C$ be a configuration that satisfies Condition~{\MixCond},
where  $C\cup \braced{\mu} \ingroundset \integers$ and $\nofdroplets{C} = n$.  
Also, let the factorization of $n$ be $n = 2^{\tau_0}p_1^{\tau_1} p_2^{\tau_2} ... p_s^{\tau_s}$, 
where $\braced{p_1,p_2,...,p_s} = \barp$ is the set of 
the odd prime factors of $n$ and $\braced{\tau_1,\tau_2,...,\tau_s}$ are their corresponding multiplicities.

If $A\ingroundset\integers$ is a configuration with $\nofdroplets{A} = n$ 
(where $n$ is as above) and $\average(A)\in\integers$, 
we will say that $A$ is \emph{$\barp$-incongruent} if $A$ is $p_r$-incongruent for all $r$.
If $A$ is $\barp$-incongruent then, by Lemma~\ref{lem: uniform wrt coprime numbers}(a),
it is $b$-incongruent for all $b$ that are powers of $p_r$'s, which,
by Corollary~\ref{cor: mc prime power factors}, implies that $A$ satisfies Condition~{\MixCond}.
Further, if $A$ is also not near-final
then Lemma~\ref{lem: mixability m = 2} implies that $\nofconcentrations{A}\ge 3$.
We summarize these observations below. (They will be often used in this section
without an explicit reference.)

\begin{observation}
	\label{obs: A mod-nouniform properties}
	Assume that a configuration $A\ingroundset \integers$ with $\average(A)\in\integers$ is $\barp$-incongruent.
	Then
	\textrm{(a)} $A$ satisfies Condition~{\MixCond}, and
	\textrm{(b)} if $A$ is not near-final then $\nofconcentrations{A}\ge 3$.
\end{observation}


\myparagraph{Proof outline}
The outline of the sufficiency proof for $n\ge 7$ is as follows (see Figure~\ref{fig: proof outline n >= 7}):
Assume that $C$ is not perfectly mixed.
Instead of dealing with $C$ directly,
we will consider a $\barp$-incongruent configuration $\intermC\ingroundset\integers$ with $\intermmu = \average(\intermC)\in\integers$
that is ``equivalent'' to $C$ in the sense
that $C$ is perfectly mixable with precision at most $1$ if and only if $\intermC$ is perfectly mixable
with precision $0$. 

It is thus sufficient to show that $\intermC$ is perfectly mixable with precision $0$.
To this end, we first apply some mixing operations to $\intermC$, producing only integer concentrations,
that convert $\intermC$ into a configuration $E$ such that:
\begin{description}
	\item{(I.0)} $E\ingroundset \integers$ and $\average(E) = \intermmu$,
	\item{(I.1)} $E$ has at least $2$ distinct non-singletons, and
	\item{(I.2)} $E$ is $\barp$-incongruent.
\end{description}

We refer to the three conditions above as Invariant~(I).
Then we show that any configuration $E$ that satisfies Invariant~(I)
has a pair of different concentrations whose mixing either preserves Invariant~(I)
or converts $E$ into a near-final configuration.
We can thus repeatedly mix such pairs, preserving Invariant~(I), until 
we produce a near-final configuration, that, by the previous section,
can be perfectly mixed with precision $0$. 

\begin{figure}[ht]
	\begin{center}
		\includegraphics[width = 4.2in]{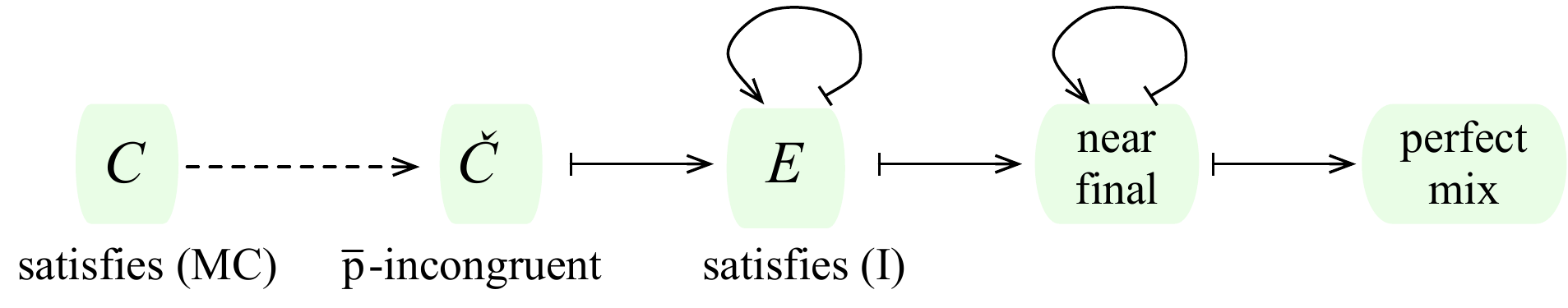}
		\caption{Proof outline for $n\ge 7$. The first dashed arrow represents
			replacing $C$ by $\intermC$. Solid arrows represent mixing operations.}
		\label{fig: proof outline n >= 7}
	\end{center}
\end{figure}


\myparagraph{Replacing $C$ by $\intermC$}
We now explain how to modify $C$. We will do it in steps.
First, let $C' = C-c_1$, for some arbitrarily chosen $c_1\in C$. 
Note that $\mu' = \average(C') = \mu-c_1 \in\integers$, that $0\in C'$, 
 and that $C'$ satisfies Condition~(MC).
By Observation~\ref{obs: offsetting does not affect reachability},
$C$ is perfectly mixable if and only if $C'$ is perfectly mixable (with the same precision), 
so it is sufficient to show that $C'$ is perfectly mixable.

Then, let $\theta\in\posintegers$ be the maximum odd integer that divides
all concentrations $c\in C'$ (that is, the greatest common odd divisor of $C'$).
Let $C'' = C'/\theta$. By Observation~\ref{obs: odd integer rescaling}(b) and the paragraph above, 
$C$ is perfectly mixable if and only $C''$ is perfectly mixable (with the same precision), so from now
on we can replace $C$ by $C''$.

By Condition~(MC) applied to $C'$,
$\theta$ is a divisor of $\mu'$, so $\mu'' = \average(C'') = \mu'/\theta\in \integers$.
Next, we claim that $C''$ is $\barp$-incongruent. To show this, we argue
by contradiction. Suppose that $C''$ is $p_r$-congruent for some $r$.
This means that there is $\beta\in\braced{0,1,...,p_r-1}$ such that
$c\equiv \beta \pmod{p_r}$ for all $c\in C''$. Since $0\in C''$ (because $0\in C'$),
we must have $\beta = 0$. In other words, all $c\in C''$ are multiples of $p_r$.
That would imply, however, that all $c\in C'$ are multiples of $\theta p_r$,
which contradicts the choice of $\theta$, completing the proof.

Finally, let $\intermC = 2\cdot C''$ and $\intermmu = 2\mu''= \average(\intermC)$. 
All concentrations in $\intermC$ are even and, since multiplying all concentrations by $2$
does not affect $\barp$-incongruence, $\intermC$ is $\barp$-incongruent. 
By Observation~\ref{obs: power of 2 rescaling}, and the properties of $C''$ established above,
$C$ is perfectly mixable with precision at most $1$ if and only if
$\intermC$ is perfectly mixable with precision $0$. Therefore, from now
on, it is sufficient to show a mixing sequence with all integral concentration values that converts
$\intermC$ into its perfect mixture $\braced{n:\intermmu}$.


\myparagraph{Converting $\intermC$ into $E$}
Let $\intermC$ be the configuration constructed above. We now show that
with at most two mixing operations, producing only integer values, we can
convert $\intermC$ into a configuration $E$ that satisfies Invariant~(I).
We start with an auxiliary lemma (Lemma~\ref{lem: at most one mod-unsafe pairs} below).

Let $A\ingroundset\integers$ be a configuration with $\average(A)\in\integers$ and $\nofdroplets{A} = n$. 
Assume that $A$ is $\barp$-incongruent.
For different concentrations $a, a' \in A$ with the same parity, we say that
the pair $(a,a')$ is \emph{$p_r$-safe} if mixing $a$ and $a'$ converts $A$
into a $p_r$-incongruent configuration; in other words, there is $a''\in A-\braced{a,a'}$
that satisfies $a'' \not\equiv \half(a+a') \pmod{p_r}$. 
(Otherwise, we say that the pair $(a,a')$ is \emph{$p_r$-unsafe}.)
We will also say that $(a,a')$ is \emph{$\barp$-safe} if it is $p_r$-safe for all $r$,
and we call it \emph{$\barp$-unsafe} otherwise.

For example, let $n=15$ and $A = \braced{11:3,10,16,18,28}$, for which $\average(A) = 7$. 
We have $\barp = \braced{p_1,p_2}$, where $p_1 = 3$ and $p_2 = 5$. 
The pair $(10,16)$ is $5$-unsafe, because mixing these two droplets produces
two droplets with concentration $13$, hence producing a $5$-congruent configuration
(all concentrations will have residue $3$ modulo $5$).
Thus the pair $(10,16)$ is $\barp$-unsafe.
All other pairs $(a,a')$ of same-parity concentrations from $A$
are $\barp$-safe.

\begin{lemma}
	\label{lem: at most one mod-unsafe pairs} 
	Let $A$ be a $\barp$-incongruent configuration with $\average(A)\in\integers$ and $\nofdroplets{A} = n$. 
	(Recall that $n\ge 7$.)
	Then
\begin{description}
	\item{\textrm{(a)}} For each $r$, there is at most one $p_r$-unsafe pair in $A$.
	\item{\textrm{(b)}}	There are at most $n-5$ droplets involved
	in same-parity concentration pairs that are $\barp$-unsafe.
	\item{\textrm{(c)}} If a concentration $a\in A$ is a non-singleton (has multiplicity at least $2$)
		then for any $b\in A$ with $b\neq a$ and the same parity as $a$, the pair $(a,b)$ is $\barp$-safe.
\end{description}
\end{lemma}

\begin{proof}
(a)
Suppose that some pair $(a_1,a_2)$ of concentrations with $a_1\neq a_2$ and same parity is $p_r$-unsafe,
and let $\beta = (\half (a_1+a_2)) \bmod p_r$. The assumption about $a_1,a_2$ implies that 
$b\equiv \beta \pmod{p_r}$ for all $b\in A- \braced{a_1,a_2}$, and the assumption that 
$A$ is $\barp$-incongruent implies that $a_i \not\equiv \beta \pmod{p_r}$ for at
least one $i\in\braced{1,2}$. We claim that this must in fact hold for
\emph{both} $i\in\braced{1,2}$. Indeed, say that $a_1 \not\equiv \beta \pmod{p_r}$
but $a_2 \equiv \beta \pmod{p_r}$. This means that
$p_r\nmid (a_1-\beta)$ and $p_r| (a_2-\beta)$, which implies that
$p_r\nmid (\half(a_1+a_2) - \beta)$, contradicting the definition of $\beta$.
Thus $a_i \not\equiv \beta \pmod{p_r}$ for both $i\in\braced{1,2}$, as claimed.

It remains to show that any other pair of concentrations is $p_r$-safe. 
Fix three arbitrary concentrations $\braced{b_1,b_2,b_3}\subseteq A - \braced{a_1,a_2}$, 
so that we have $b_j\equiv\beta\pmod{p_r}$ for $j\in\braced{1,2,3}$.
Consider any two different same-parity concentrations $c_1,c_2\in A$ with $\braced{c_1,c_2}\neq\braced{a_1,a_2}$, and
let $A'$ be obtained from $A$ by mixing droplets $c_1$ and $c_2$.
Then $A'$ must still contain some droplet $b_j$ and, since $\braced{c_1,c_2}\neq\braced{a_1,a_2}$, 
$A'$ will also contain some droplet $a_i$. 
As we have $a_i\not\equiv b_j \pmod{p_r}$, $A'$ is $p_r$-incongruent, 
and thus $(c_1,c_2)$ is $p_r$-safe.
	
(b) By part~(a),
	the	number of concentrations involved in same-parity $\barp$-unsafe pairs is at most
	$2s$, where $s$ is the number of distinct odd prime factors of $n$,
	so it remains to show that  $2s \leq n-5$.
	Indeed,	if $n$ equals either $7$ or $8$ (for which $s=1$ or $0$, respectively), then the inequality holds. 
	For $n\ge 9$, using the fact that $s\le \log_3{n}$, it is sufficient to show that
	$2\log_3{n} \le n-5$. This is true, because for $n=9$ the equality holds,
	and for $n\ge 9$ the left-hand side grows slower than the right-hand side.	
	
(c) Fix some factor $p_r$ of $n$. As $A$ is $p_r$-incongruent, there is
a concentration $c\in A$ with $c\not\equiv a \pmod{p_r}$. We have two cases. 
If $b\equiv a\pmod{p_r}$ then $b\neq c$, so after mixing the new configuration
$A'$ will contain $c$ and $\half(a+b)$, where $\half(a+b) \equiv a \pmod{p_r}$,
so $c\not\equiv \half(a+b) \pmod{p_r}$.
On the other hand, if $b\not\equiv a\pmod{p_r}$, then $A'$ will contain $a$ and $\half(a+b)$,
and $a\not\equiv \half(a+b) \pmod{p_r}$. Thus $(a,b)$ is $p_r$-safe.
As this holds for all $r$, $(a,b)$ is $\barp$-safe.
\end{proof}

The configuration $\intermC$ constructed earlier contains only even concentration values,
already satisfies $\intermC\cup\braced{\intermmu}\ingroundset\integers$ and is $\barp$-incongruent 
(that is, it satisfies conditions~(I.0) and~(I.2) for $E$).  
It remains to show that there are mixing operations involving only droplets 
already present in $\intermC$ (and thus of even value, to assure that 
Condition~(I.0) holds) that preserve condition~(I.2), and such that the resulting
configuration $E$ satisfies condition~(I.1).
If $\intermC$ already has two or more non-singletons, we can take $E = \intermC$ and we are done, 
so assume otherwise, namely that there is either exactly one non-singleton in $\intermC$ or none.
We consider three cases.

\medskip\noindent
{\mycase{1}} $\intermC$ has one non-singleton $a$ and its multiplicity is $f\ge 3$.
Mix $a$ with any singleton $b$ and let $E$ be the resulting configuration.
In $E$ we have two non-singletons and condition~(I.2) will be satisfied, 
by Lemma~\ref{lem: at most one mod-unsafe pairs}(c). Thus $E$ satisfies Invariant~(I).

\medskip\noindent
{\mycase{2}} $\intermC$ has one non-singleton $a$ and its multiplicity is $2$.
By 	Lemma~\ref{lem: at most one mod-unsafe pairs}(b),
there are at least $5$ droplets in $\intermC$ not involved in any
$\barp$-unsafe pair. Thus there are at least $3$ singletons, say $b,c,d$,
that are not involved in any $\barp$-unsafe pair.
Mixing one of pairs $(b,c)$ or $(b,d)$ produces a concentration other than
$a$. Mix this pair, and let $E$ be the resulting configuration.
Then $E$ satisfies Invariant~(I).

\medskip\noindent
{\mycase{3}} $\intermC$ has only singletons.
By Lemma~\ref{lem: at most one mod-unsafe pairs}(b),
there is a singleton, say $b\in \intermC$, that is not involved in any
$\barp$-unsafe pair (in fact, there are at least five, but we need just one here).
Let $c\in \intermC-\braced{b}$ be a singleton nearest to $b$, that is one
that minimizes $|c-b|$. By the choice of $b$, the pair $(b,c)$ is
$\barp$-safe. Let $\intermC' = \intermC - \braced{b,c} \cup\braced{a,a}$,
for $a=\half(b+c)$, be the configuration obtained by mixing this pair.
$\intermC'$ is $\barp$-incongruent and
in $\intermC'$ we have only one non-singleton $a$ and its multiplicity is $2$.
We can thus apply Case~2 above to $\intermC'$, converting it to $E$.
(Note that, unlike for the $\intermC$ in Case~2, our $a$ may be odd. But since we
do not mix $a$ in Case~2, the argument is still valid.)


\myparagraph{Preserving Invariant~(I)}
We now present the last part of the proof, following the outline given at the beginning
of this section. Let $E\ingroundset\integers$ be the configuration, say
$E = \braced{f_1:e_1, f_2:e_2, ... , f_m:e_m}$, with $\average(E)=\intermmu$,
obtained from $\intermC$ by a sequence of mixing operations, as described earlier.
If $E$ is near-final then $E$ has a perfect-mixing sequence 
by Lemma~\ref{lem: mixability for near-final}.
Otherwise, we show that $E$ has a pair of concentrations 
whose mixing produces a configuration that either preserves Invariant~(I) or is near-final.

Let $e_i,e_j\in E$ be two different concentrations. Let $e = \half(e_i+e_j)$, and denote by
$E' = E- \braced{e_i,e_j}\cup \braced{e,e}$ the configuration obtained from $E$ by mixing $e_i$ and $e_j$.
The following two notions will be useful in our analysis:
\begin{itemize}
	\item  Pair $(e_i,e_j)$ will be called \emph{(I)-safe} if $E'$ satisfies Invariant~(I).
	We will be always choosing $e_i$ and $e_j$ with the same parity, which is a sufficient
	and necessary condition for $E'$ to satisfy condition~(I.0). Also,
	for $E'$ to satisfy condition~(I.2), the pair $(e_i,e_j)$ must be $\barp$-safe.
	\item Pair $(e_i,e_j)$ will be called \emph{near-final} if $E'$ is near-final.
	Note that it is possible for $(e_i,e_j)$ to be both (I)-safe and near-final.
\end{itemize}
We next prove that if configuration $E$ satisfies Invariant~(I) 
then it must contain a pair of different concentrations that is either (I)-safe
or near-final. This will show that we can repeatedly mix $E$, 
maintaining Invariant~(I), until we turn $E$ into a near-final configuration,
which we can then perfectly mix using Lemma~\ref{lem: mixability for near-final}.

\begin{lemma}
	\label{lem: invariant, two same parity, fi >= 3}
	Assume that $E$ contains two different concentrations $e_i,e_j\in E$
	with the same parity and $f_i\ge 3$.
	If $E$ satisfies Invariant~(I) then the pair $(e_i,e_j)$ is (I)-safe.
\end{lemma}

\begin{proof}
	Let $e = \half(e_i+e_j)$ and let $E' = E-\braced{e_i,e_j}\cup\braced{e,e}$ be obtained from
	mixing $e_i$ and $e_j$. 	
	Since $f_i > 1$, Lemma~\ref{lem: at most one mod-unsafe pairs}(c) implies that
	condition (I.2) holds for $E'$.	In $E'$ we will still have at least two droplets of
	concentration $e_i$ and at least two droplets of concentration $e\neq e_i$.
	So condition (I.1) holds as well.
	(We remark that we could end up with $\nofconcentrations{E'} = 2$, which can happen if $f_j = 1$
	and $\nofconcentrations{E}=3$ with $e\in E$.
	If so, since $E'$ satisfies (I.2), it must also satisfy condition~{\MixCond}, and
	therefore, by Lemma~\ref{lem: mixability m = 2}, in this case $E'$ is actually near-final;
	that is, $(e_i,e_j)$ is a near-final pair.)
\end{proof}

\begin{lemma}
	\label{lem: invariant, three same parity, fi, fj >= 2}
	Assume that $E$ contains three different concentrations $e_i,e_j, e_k\in E$
	with the same parity and $f_i, f_j\ge 2$.
	If $E$ satisfies Invariant~(I) then one of $(e_i,e_k)$, $(e_j,e_k)$ is (I)-safe.
\end{lemma}

\begin{proof}
	Without loss of generality, assume that $|e_i-e_k| \le |e_j-e_k|$ (otherwise
	swap $i$ and $j$). We show that $(e_i,e_k)$ is (I)-safe.
	Let $e = \half(e_i+e_k)$ and let $E' = E-\braced{e_i,e_k}\cup\braced{e,e}$ be obtained from
	mixing $e_i$ and $e_k$. 
	Since $f_i > 1$, Lemma~\ref{lem: at most one mod-unsafe pairs}(c) implies that
	condition (I.2) holds for $E'$.
	From $|e_i-e_k| \le |e_j-e_k|$, we have that $e\neq e_j$.
	So in $E'$ we will have at least two droplets of
	concentration $e_j$ and at least two droplets of concentration $e\neq e_j$.
	This means that condition (I.1) holds as well.
\end{proof}

\begin{lemma}
	\label{lem: invariant, m >= 4, three same parity, fi >= 2, fj = fk = 1}
	Assume that $\nofconcentrations{E}\ge 4$ and that
	 $E$ contains three different concentrations $e_i,e_j, e_k\in E$
	with the same parity such that $f_i\ge 2$ and $f_j = f_k = 1$.
	If $E$ satisfies Invariant~(I),
	then one of $(e_i,e_j)$, $(e_i,e_k)$ is (I)-safe.
\end{lemma}

\begin{proof}
	By condition (I.1), there is another concentration $e_l\in E - \braced{e_i,e_j,e_k}$
	with $f_l\ge 2$.
	Without loss of generality, we can assume that $e = \half(e_i+e_j)\neq e_l$ (otherwise we can
	use $e_k$ instead of $e_j$).
	Mixing $e_i$ and $e_j$ produces $E' = E - \braced{e_i,e_j}\cup \braced{e,e}$.
	Since $f_i\ge 2$, condition (I.2) is satisfied.
	In $E'$ there are at least two droplets with concentration $e_l$
	and at least two droplets with concentration $e\neq e_l$, so
	condition (I.1) is satisfied as well.
\end{proof}

\begin{lemma}
	\label{lem: m = 3, preserves invariant}
	Assume that $\nofconcentrations{E}=3$.
	If $E$ satisfies Invariant~(I) then $E$ has a pair of concentrations
	that is either (I)-safe or near-final.
\end{lemma}

\begin{proof}
	Let $E = \braced{f_1:e_1, f_2:e_2, f_3:e_3}$. Reorder $E$ so that
	$f_1\geq f_2 \geq f_3$. From $f_1+f_2+f_3 = n \ge 7$
	we have that $f_1\ge 3$ and $f_2\ge 2$. By symmetry, we can
	also assume that $e_1$ is even. 
	If either $e_2$ or $e_3$ is even, then the existence of an (I)-safe pair follows from
	Lemma~\ref{lem: invariant, two same parity, fi >= 3}.
	So we can assume that $e_2,e_3$ are odd.
	
	Let $e = \half(e_2+e_3)$ and let $E' = E-\braced{e_2,e_3}\cup\braced{e,e}$ 
	be obtained from mixing $e_2$ and $e_3$. 	
	Since $f_2 \ge 2$, Lemma~\ref{lem: at most one mod-unsafe pairs}(c) implies that
	condition (I.2) holds for $E'$.
	This, and Observation~\ref{obs: A mod-nouniform properties}(a) imply that if 
	$\nofconcentrations{E'} = 2$ then, by Lemma~\ref{lem: mixability m = 2},
	$E'$ is near-final and thus $(e_2,e_3)$ is a near-final pair.
	So for the rest of the proof we assume that $\nofconcentrations{E'}\ge 3$.
	(For $(e_2,e_3)$ to be (I)-safe, it is now sufficient to prove that $E'$ satisfies (I.1).)
	
	If $e\neq e_1$, in $E'$
	we have at least three droplets with concentration $e_1$ and at least
	two with concentration $e$, so $E'$ satisfies (I.1).
	Otherwise, $e = e_1$. 
	Now, $f_2 = f_3$ implies that $E'$ is near-final 
	(by partitioning $E'$ into singletons $\braced{e_1}$ and pairs $\braced{e_2,e_3}$),
	and thus $(e_2,e_3)$ is a near-final pair.
	Instead, assume that $f_2 > f_3$.
	As $\nofconcentrations{E'}\ge 3$ (and $e = e_1$), $f_3 \ge 2$ and thus $f_2\ge 3$. 
	This implies that in $E'$ there are at least five droplets with concentration $e_1$ and at least
	two droplets with concentration $e_2$, so $E'$ satisfies (I.1).
\end{proof}

\begin{lemma}
	\label{lem: m = 4, preserves invariant}
	Assume that $\nofconcentrations{E}=4$.
	If $E$ satisfies Invariant~(I) then $E$ has an (I)-safe pair.
\end{lemma}

\begin{proof}
	Let $E = \braced{f_1:e_1, f_2:e_2, f_3:e_3, f_4:e_4}$.
	By symmetry and reordering, respectively, we can assume that $e_1$ is even and that
	$f_1\ge f_2\ge f_3\ge f_4$. This, and condition~(I.1) imply that $f_1\geq f_2\geq 2$.
	We consider two cases, depending on the value of $f_1$.
	
	\smallskip\noindent
	\mycase{1} $f_1\ge 3$.
	If at least one of $e_2,e_3,e_4$ is even, then the existence of an
	(I)-safe pair follows from Lemma~\ref{lem: invariant, two same parity, fi >= 3}.
	
	So assume now that $e_2,e_3,e_4$ are all odd.
	If $f_3\ge 2$, we obtain an (I)-safe pair from
	Lemma~\ref{lem: invariant, three same parity, fi, fj >= 2}.
	Otherwise, $f_3 = f_4 = 1$, and we obtain an (I)-safe pair
	from Lemma~\ref{lem: invariant, m >= 4, three same parity, fi >= 2, fj = fk = 1}.
	
	\smallskip\noindent
	\mycase{2} $f_1 = 2$. Then $n\ge 7$ implies that $f_2 = f_3 = 2$ as well.
	If two concentrations among $e_2,e_3,e_4$ are even,
	or if $e_2,e_3,e_4$ are all odd,
	the existence of an (I)-safe pair follows from
	Lemma~\ref{lem: invariant, three same parity, fi, fj >= 2}.
	
	Otherwise, one of $e_2,e_3,e_4$ is even and two are odd.
	We then want to mix $e_4$ with the one of $e_1,e_2,e_3$ that has the same parity as $e_4$.
	For concreteness, assume that $e_2$ is even and $e_3,e_4$ are odd.
	(The argument in all other cases is the same.)
	We claim that $(e_3,e_4)$ is (I)-safe. 
	Indeed, let $E'= E-\braced{e_3,e_4}\cup \braced{e,e}$,
	for $e = \half (e_3+e_4)$.
	Since $f_3>1$, Lemma~\ref{lem: at most one mod-unsafe pairs}(c)
	implies that condition~(I.2) holds.
	In $E'$ we have at least two droplets with concentration $e$.
	If $e\neq e_1$, then $E'$ has
	two droplets with concentration $e_1\neq e$;
	otherwise, if $e = e_1$, then $E'$ has
	two droplets with concentration $e_2\neq e$.
	Thus condition~(I.1) holds for $E'$.
\end{proof}

\begin{lemma}
	\label{lem: m => 5, preserves invariant}
	Assume that $\nofconcentrations{E}\ge 5$.
	If $E$ satisfies Invariant~(I) then $E$ has an (I)-safe pair.
\end{lemma}

\begin{proof}
Let $E = \braced{f_1:e_1, f_2:e_2, ... , f_m:e_m}$, for $m\ge 5$.
	By symmetry and reordering, respectively, we can assume that $e_1$ is even and that
	$f_i\ge f_{i+1}$, for $i=1,2,....,m-1$. By condition~(I.1), we have $f_1, f_2\ge 2$.
	We consider several cases.
	
	\smallskip\noindent
	\mycase{1} $f_1\ge 3$. The same	argument as in Case~1 in the proof of
			Lemma~\ref{lem: m = 4, preserves invariant} applies here.
	
	\smallskip\noindent
	\mycase{2} $f_1 = 2$. Then $f_2 = 2$ as well. We have some sub-cases.
	
	\vspace{-0.1in}
	
	\begin{description}
		\item{\mycase{2.1}} $f_3 = 2$. In this case, choose three concentrations
			among $e_1,e_2,e_3,e_4,e_5$ with the same parity. This will give us
 			three concentrations $e_i,e_j,e_k$ with the same parity and with
			$f_i = 2$. If $f_j=2$ or $f_k = 2$,
			we obtain an (I)-safe pair from Lemma~\ref{lem: invariant, three same parity, fi, fj >= 2},
			otherwise we obtain an (I)-safe pair
			from Lemma~\ref{lem: invariant, m >= 4, three same parity, fi >= 2, fj = fk = 1}.
			
		\item{\mycase{2.2}}	$f_3=...=f_m = 1$ and $e_2$ is odd. Among $e_3,e_4,e_5$ there
			are either two even or two odd concentrations. By symmetry, we can assume
			$e_3,e_4$ are even.	This gives us three concentrations $e_i,e_j,e_k$ that
			satisfy the assumptions of 
			Lemma~\ref{lem: invariant, m >= 4, three same parity, fi >= 2, fj = fk = 1},
			so we obtain an (I)-safe pair by applying this lemma.
		
		 \item{\mycase{2.3}} $f_3=...=f_m = 1$ and $e_2$ is even.
		 	If any concentration among $e_3,e_4,...,e_m$ is even, then we obtain an (I)-safe pair
		 	from Lemma~$\ref{lem: invariant, three same parity, fi, fj >= 2}$.
			Otherwise, $e_3, e_4, \ldots, e_m$ are all odd. This set has $m-2 = n-4$ droplets.
			Thus, by Lemma~\ref{lem: at most one mod-unsafe pairs}(b), 
			there is at least one concentration $e_i$, for $i\in\braced{3,4,...,m}$, 
			such that any pair $(e_i,e_j)$,
			for $j\in\braced{3,4,...,m}-\braced{i}$, is $\barp$-safe.
			Let $E' = E - \braced{e_i,e_j}\cup \braced{e,e}$ be obtained from mixing $e_i$ and $e_j$,
			for $e = \half(e_i+e_j)$.
			By the choice of $e_i$, $E'$ satisfies (I.2).
			$E'$ satisfies (I.1) because it has at least two droplets with concentration $e_1$ and 
			at least two droplets with concentration $e_2$.
	\end{description}	
\end{proof}


\myparagraph{Completing the proof}
We are now ready to complete the proof that Condition~{\MixCond} in 
Theorem~\ref{thm: miscibility characterization}(a) is sufficient 
for perfect mixability when $n\geq 7$. The argument follows
the outline given at the beginning of this section and depicted in
Figure~\ref{fig: proof outline n >= 7}.

Assume that $C$ satisfies Condition~{\MixCond}. 
If $C$ is already perfectly mixed then we are done.
Otherwise, as described earlier in the proof, 
we first replace $C$ by configuration $\intermC \ingroundset\integers$
such that (i) $\intermmu = \average(\intermC)\in\integers$,
all values in $\intermC$ are even, and $\intermC$ is $\barp$-incongruent,
and (ii) $C$ is perfectly mixable with precision at most $1$ if and only if 
$\intermC$ is perfectly mixable with precision $0$. 

Then we show that $\intermC$ has a perfect-mixing sequence (with precision $0$), converting
$\intermC$ into its perfect mixture $\braced{n:\intermmu}$.
To this end, we first perform some mixing operations (at most two)
that convert $\intermC$ into a configuration $E$
that either satisfies Invariant~(I) or is near-final.
If this $E$ is near-final, we can complete the mixing
sequence using Lemma~\ref{lem: mixability for near-final}.
If this $E$ is not near-final, then condition~(I.2) implies that $E$ satisfies Condition~{\MixCond}
which, in turn, by Lemma~\ref{lem: mixability m = 2}, implies that $\nofconcentrations{E}\ge 3$.
Therefore, depending on the value of $\nofconcentrations{E}$, we can apply
one of Lemmas~\ref{lem: m = 3, preserves invariant}, \ref{lem: m = 4, preserves invariant},
or \ref{lem: m => 5, preserves invariant}, 
to show that $E$ has a pair of concentrations that is either (I)-safe or near-final.
We can thus apply the above argument repeatedly to $E$.
As in Section~\ref{subsec: near-final configurations}, 
each mixing decreases the value of $\potential(E) = \sum_{e\in E}(e - \intermmu)^2$.
Thus after a finite number of steps
we eventually convert $E$ into a near-final configuration
(as in the cases for Lemma~\ref{lem: m = 3, preserves invariant}), 
that has a mixing sequence by Lemma~\ref{lem: mixability for near-final}.


\subsection{Proof for $n=5$}
\label{subsec: proof for n equal 5}



In this sub-section we prove that Condition~$\MixCond$ in 
Theorem~\ref{thm: miscibility characterization}(a) is sufficient 
for perfect mixability when $n=5$.
The overall argument is similar to the case $n \geq 7$ we considered in 
Section~\ref{subsec: proof for n greater-equal than 7}
(and depicted in Figure~\ref{fig: proof outline n >= 7}),
although this time we need a slightly different invariant. 
This is because in the case when $n=5$ there are configurations that
satisfy Invariant~(I) but do not contain any pair of concentrations
whose mixing preserves Invariant~(I).
For example, $E = \braced{0,0,4,4,7}$ with $\average(E)=3$
satisfies Invariant~(I). The only pair of different concentrations with
the same parity is $(0,4)$; however, after mixing these concentrations,
the new configuration will violate condition~(I.1).

Let $A\ingroundset\integers$ be a configuration with $n = \nofdroplets{A} = 5$ and
 $\average(A)\in\integers$.
We say that $A$ is \emph{blocking} if $A=\braced{3:a_1,a_2,a_3}$
where $a_1\neq\half(a_2+a_3)$ and $a_1$ has parity different than $a_2,a_3$.
Otherwise we say that $A$ is \emph{non-blocking}.
For example, $A = \braced{0,0,0,3,7}$, with $\average(A)= 2$, is blocking.
The intuition is that this $A$ has only one pair of same-parity different
concentrations, namely $(3,7)$, but this pair is not $5$-safe --
mixing $3$ and $7$ produces configuration
$A' = \braced{0,0,0,5,5}$ that is $5$-congruent (in fact, it also violates Condition~{\MixCond}).

So, assume that we are given a configuration $C$ with $n = \nofdroplets{C} = 5$ and
$C \cup \braced{\mu} \ingroundset \integers$, that
satisfies Condition~{\MixCond}. 
As in Section~\ref{subsec: proof for n greater-equal than 7},
if $C$ is already perfectly mixed then we are done.
Otherwise we start by converting $C$ into a configuration
 $\intermC\ingroundset\integers$, with $\intermmu = \average(\intermC)\in\integers$,
such that (i) all concentrations in $\intermC$ are even,
(ii) $\intermC$ is $5$-incongruent, and
(iii) $C$ is perfectly mixable with precision at most $1$ if and only if $\intermC$
is perfectly mixable with precision $0$.

We then simply take $E = \intermC$ (unlike in Section~\ref{subsec: proof for n greater-equal than 7}, we 
don't need to modify $\intermC$). It thus remains to show that 
$E$ is perfectly mixable with precision $0$. 
In order to do so, we will have $E$ maintain the following Invariant~(I'):
\begin{description}
	\item{(I.0)} $E\ingroundset \integers$ and $\average(E) = \intermmu$,
	\item{(I.1')} $E$ is non-blocking, and
	\item{(I.2)} $E$ is $5$-incongruent.
\end{description}
By the properties of $\intermC$, the initial set $E$ satisfies conditions (I.0) and (I.2) and, since
all concentrations in $E$ are even, it also satisfies (I.1'). Thus
$E$ satisfies Invariant~(I') and the rest of the proof is devoted to constructing a
sequence of mixing operations that preserve Invariant~(I') until $E$ becomes near-final,
which can be mixed perfectly using Lemma~$\ref{lem: mixability for near-final}$.

At this point, we observe that, although in the previous section we considered the case $n \ge 7$,
the claims in Lemma~$\ref{lem: at most one mod-unsafe pairs}$(a) and~(c)
hold also for $n=5$ (with $\barp = \braced{5}$), and we will be using them in the proof.
(Lemma~$\ref{lem: at most one mod-unsafe pairs}$(b) does not apply to $n=5$, however.)
We will also frequently use Observation~$\ref{obs: A mod-nouniform properties}$ (with $\barp = \braced{5}$)
that follows directly from Lemma~\ref{lem: mixability m = 2}.


\myparagraph{Preserving Invariant~(I')}
Assume that $E$ satisfies Invariant~(I'). We say that a pair of distinct concentrations in $E$ is \emph{(I')-safe}
if the configuration obtained from $E$ by mixing these concentrations satisfies Invariant~(I').
(This is an analogue of the notion of (I)-safe pairs, introduced in Section~\ref{subsec: proof for n greater-equal than 7}.)
We now show that each configuration $E$ that satisfies Invariant~(I') has a pair
of concentrations that is either (I')-safe or near-final.
As we will always choose a pair of concentrations with the same parity for mixing,
condition~(I.0) will be trivially preserved, so in the proofs below,
to show that a pair is (I')-safe, we will focus on
explaining why the other two conditions are preserved.

\begin{lemma}
	\label{lem: m = 3, preserves invariant, n = 5}
	Assume that $\nofconcentrations{E}=3$.
	If $E$ satisfies Invariant~(I') then $E$ has a pair of concentrations
	that is either (I')-safe or near-final.
\end{lemma}

\begin{proof}
	Let $E = \braced{f_1:e_1, f_2:e_2, f_3:e_3}$.
	By symmetry and reordering, respectively, we can assume that $e_1$ is even
	and that $f_1\geq f_2 \geq f_3$.
	We analyze two cases based on $f_1$'s value:
	
	\smallskip\noindent
	\mycase{1} $f_1 = 2$. Thus, $f_2 = 2$ and $f_3 = 1$. We consider two sub-cases.

\vspace{-0.05in}
	
	\begin{description}
	
	\item{\mycase{1.1}} $e_2$ is even.
	We will mix $e_1$ and $e_2$, producing $E'=\braced{e_1,e_2,e_3,2:e}$,
	for $e=\half(e_1+e_2)$. Obviously, $e\notin\braced{e_1,e_2}$. 
	If $e=e_3$ then $E'$ is near-final
	(using the partition of $E'$ into $\braced{e_1,e_2},\braced{e_3},\braced{e_3},\braced{e_3}$)
	and thus $(e_1,e_2)$ is a near-final pair.
	So, assume that $e\neq e_3$.	
    
	Since $f_1 = 2$, by Lemma~$\ref{lem: at most one mod-unsafe pairs}$(c), $E'$
	satisfies condition~(I.2). Further, as $\nofconcentrations{E'} = 4$,
	$E'$ also satisfies condition~(I.1').
	Therefore, $(e_1,e_2)$ is indeed (I')-safe.
	
	\item{\mycase{1.2}} $e_2$ is odd.
	Without loss of generality we can assume that $e_3$ is even (by the odd-even symmetry
	between $e_1$ and $e_2$.)
	We mix $e_1$ and $e_3$, and let $E'$ be the resulting configuration.
	Since $f_1 = 2$, by Lemma~$\ref{lem: at most one mod-unsafe pairs}$(c), 
	pair $(e_1,e_3)$ is $5$-safe, so $E'$ satisfies condition~(I.2).
	This, together with Lemma~\ref{lem: mixability m = 2} (as $n = 5$ is prime),
	implies that $\nofconcentrations{E'} > 2$, which means that
	$\half(e_1+e_3)\neq e_2$,
	implying in turn that $E'$ has two non-singletons.
	Thereby, $E'$ satisfies~(I.1') and thus we can conclude that pair $(e_1,e_3)$ is (I')-safe.
	
	\end{description}
	
	\smallskip\noindent
	\mycase{2} $f_1 = 3$. Thus, $f_2 = f_3 = 1$.
	Since $E$ satisfies~(I.1')
	we have that at least one of $e_2$ and $e_3$ is even. 
	By symmetry, we can assume that $e_2$ is even. We mix $e_1$ and $e_2$,
	and let $E'$ be the new configuration.
	As $f_1 = 3$, by Lemma~$\ref{lem: at most one mod-unsafe pairs}$(c), $E'$ satisfies~(I.2).
	Further, $\half(e_1+e_2)\neq e_3$, because otherwise we would have $\nofconcentrations{E'} = 2$,
	contradicting Lemma~\ref{lem: mixability m = 2}.
	So $E'$ contains two non-singletons, and thus it satisfies condition~(I.1').
	Therefore, pair $(e_1,e_2)$ is (I')-safe.	
\end{proof}


\begin{lemma}
	\label{lem: m = 4, preserves invariant, n = 5}
	Assume that $\nofconcentrations{E}=4$.
	If $E$ satisfies Invariant~(I') then $E$ has an (I')-safe pair.
\end{lemma}

\begin{proof}
	Let $E = \braced{f_1:e_1, f_2:e_2, f_3:e_3, f_4:e_4}$.
	By symmetry and reordering, respectively, we can assume that $e_1$ is even
	and that $f_1 = 2$ and $f_2 = f_3 = f_4 = 1$. We analyze three cases based
	on the parities of $e_2,e_3,e_4$.
	
	\smallskip\noindent
	\mycase{1} At least two of $e_2,e_3,e_4$ are even.
	Assume without loss of generality that $e_2,e_3$ are even.
	Let $e = \half (e_1+e_2)$. Further, assume that $e \neq e_4$
	(for otherwise we can use $e_3$ instead of $e_2$).
	Let $E'  = E - \braced{e_1,e_2} \cup \braced{e,e}$ be obtained by mixing $e_1$ and $e_2$.
	As $f_1 = 2$, $E'$ satisfies~(I.2) by Lemma~$\ref{lem: at most one mod-unsafe pairs}$(c).
	If $e\neq e_3$ then $\nofconcentrations{E'} = 4$, so $E'$ satisfies~(I.1').
	If $e = e_3$ then $E' = \braced{e_1,3:e_3,e_4}$ and $e_1,e_3$ are even, so 
	$E'$ satisfies~(I.1') as well. Thus, pair $(e_1,e_2)$ is (I')-safe.
	
	\smallskip\noindent
	\mycase{2} Exactly one of $e_2,e_3,e_4$ is even.
	Assume without loss of generality that $e_2$ is even. Let $e=\half(e_1+e_2)$, and
	let $E' = E - \braced{e_1,e_2} \cup \braced{e,e}$ be obtained by mixing $e_1$ and $e_2$.
	$E'$ satisfies~(I.2) by Lemma~$\ref{lem: at most one mod-unsafe pairs}$(c).
	If $e\in E$ then $e\in\braced{e_3,e_4}$, where both $e_3,e_4$ are odd,
	and one of $e_3,e_4$ is a non-singleton in $E'$.
	Otherwise, $e\notin E$ and thus $\nofconcentrations{E'} = 4$.
	Either way $E'$ satisfies~(I.1') and thus pair $(e_1,e_2)$ is (I')-safe.

	\smallskip\noindent
	\mycase{3} $e_2,e_3,e_4$ are all odd.
	By Lemma~$\ref{lem: at most one mod-unsafe pairs}$(a)
	there is at most one $5$-unsafe pair. Assume that $(e_2,e_3)$ is $5$-safe
	(otherwise use $e_4$ instead of $e_3$).
	Thus mixing $e_2$ and $e_3$ produces $E' = E - \braced{e_2,e_3} \cup \braced{e,e}$,
	for $e = \half(e_2+e_3)$, that satisfies condition~(I.2).
	Moreover, $e\in E$ would imply $\nofconcentrations{E'} = 2$
	(contradicting Lemma~$\ref{lem: mixability m = 2}$),
	so we must have $e \notin E$, and thus both $e_1$ and $e$ 
	are non-singletons in $E'$.
	Therefore, $E'$ preserves~(I.1') and pair $(e_2,e_3)$ is (I')-safe.
\end{proof}


\begin{lemma}
	\label{lem: m = 5, preserves invariant, n = 5}
	Assume that $\nofconcentrations{E}=5$.
	If $E$ satisfies Invariant~(I') then $E$ has an (I')-safe pair.
\end{lemma}

\begin{proof}
Let $E = \braced{e_1, e_2,e_3,e_4,e_5}$.
By symmetry and reordering we can assume that $e_1,e_2,e_3$ have the same parity, say even.
Additionally, by Lemma~$\ref{lem: at most one mod-unsafe pairs}$(a), there is at most one
$5$-unsafe pair in $E$, so we can assume that it does not involve $e_1$.
In other words, all pairs involving $e_1$ and any other even concentration are $5$-safe.  
We now have three cases, given below. In each case we mix $e_1$ with some other even concentration,
so conditions~(I.0) and~(I.2) will be satisfied, and we only need to ensure that
(I.1') is satisfied as well in order to show an (I')-safe pair in $E$.

\smallskip\noindent
\mycase{1} $e_4$ and $e_5$ are even.
Choose $i\neq 1$ for which $|e_1 - e_i|$ is minimized, and let
$e = \half (e_1 + e_i)$. Then $e\notin E$.
Mixing $e_1$ and $e_i$ gives us configuration $E' = E - \braced{e_1,e_i}\cup \braced{e,e}$
with  $\nofconcentrations{E'}=4$, so $E'$ satisfies~(I.1').

\smallskip\noindent
\mycase{2} $e_4$ and $e_5$ have different parity. Say that $e_4$ is even and $e_5$ is odd.
Let $e =\half (e_1 + e_2)$ and, without loss of generality, 
assume that $e\neq e_5$ (otherwise, use $e_3$ instead of $e_2$).  
Mixing $e_1$ and $e_2$ gives us configuration $E' = E - \braced{e_1,e_2}\cup \braced{e,e}$.
If $e\notin E$, then $\nofconcentrations{E'}=4$.
Otherwise, $e\in\braced{e_3,e_4}$ and, as $e_3,e_4$ are both even, $E'$ is non-blocking.
Thus in both sub-cases $E'$ satisfies~(I.1').

\smallskip\noindent
\mycase{3} $e_4$ and $e_5$ are odd.
Without loss of generality, assume that $|e_1 - e_2|\le |e_1-e_3|$.
Then $e =\half (e_1 + e_2) \neq e_3$. Let $E'= E - \braced{e_1,e_2}\cup \braced{e,e}$ be obtained by mixing $e_1$ and $e_2$.
If $e\notin E$, then $\nofconcentrations{E'}=4$.
Otherwise, $e\in \braced{e_4,e_5}$ and, as $e_4,e_5$ are both odd, $E'$ is non-blocking. 
Thus in both sub-cases $E'$ satisfies~(I.1').   
\end{proof}


\myparagraph{Completing the proof}
We can now prove that Condition~{\MixCond} in 
Theorem~\ref{thm: miscibility characterization}(a) is sufficient 
for perfect mixability when $n=5$.
Assume that $C$ satisfies Condition~{\MixCond}. 
If $C$ is perfectly mixed, then we are done.
Otherwise, as described earlier in this section, 
we convert $C$ into configuration $\intermC\ingroundset\integers$
such that $C$ is perfectly mixable with precision at most $1$ if and only if
$\intermC$ is perfectly mixable with precision $0$.
This $\intermC$ is $5$-incongruent, all its concentrations are even, and
it satisfies $\intermmu = \average(\intermC)\in\integers$.
Thus, if we take $E  = \intermC$, this $E$ satisfies Invariant~(I').

If $E$ is near-final, then we use Lemma~\ref{lem: mixability for near-final}
to perfectly mix $E$.
If this $E$ is not near-final, then, depending on the value of $\nofconcentrations{E}$, 
we apply one of Lemmas~\ref{lem: m = 3, preserves invariant, n = 5}, 
\ref{lem: m = 4, preserves invariant, n = 5},
or \ref{lem: m = 5, preserves invariant, n = 5}, 
to show that $E$ has a pair of concentrations that is either (I')-safe or near-final.
As in Section~\ref{subsec: proof for n greater-equal than 7},
the value of $\potential(E)$ decreases at least by $1$ after each mixing operation.
So after a finite sequence of mixing operations involving (I')-safe pairs,
$E$ must become near-final (as in the cases for Lemma~\ref{lem: m = 3, preserves invariant, n = 5}), 
which we then perfectly mix using Lemma~\ref{lem: mixability for near-final}.


\subsection{Proof for $n=6$}
\label{subsec: proof for n equal 6}



We now sketch the proof that Condition~$\MixCond$ in 
Theorem~\ref{thm: miscibility characterization}(a) is sufficient 
for perfect mixability when $n=6$. This proof can be obtained by
modifying the proof for $n=5$ in Section~$\ref{subsec: proof for n equal 5}$.
This modification takes advantage of the fact that the proof for $n=5$
relies on $n=5$ having only one odd prime factor, which is also true for $n=6$.
As for $n=5$, Lemma~$\ref{lem: at most one mod-unsafe pairs}$(a) and~(c)
hold also for $n=6$ (with $\barp = \braced{3}$).

The overall structure of the argument is identical.
Let $C\cup \braced{\mu} \ingroundset \integers$ be the given configuration
satisfying Condition~{\MixCond}.
If $C$ is perfectly mixed, we are done.
Otherwise, construct $\intermC$ as in Section~\ref{subsec: proof for n greater-equal than 7}
and take $E = \intermC$. We then repeatedly mix (I')-safe pairs in $E$ until it becomes near-final,
which can be perfectly mixed using Lemma~\ref{lem: mixability for near-final}.

The definition of Invariant~(I') (and thus the corresponding concept of (I')-safe pairs) 
is adjusted in a natural way:
Condition~(I.2) now requires $E$ to be 3-incongruent.
Further, a configuration $A$ is now called \emph{blocking} if $A=\braced{4:a_1,a_2,a_3}$
where $a_1\neq\half(a_2+a_3)$ and $a_1$ has parity different than $a_2,a_3$.

Given a configuration $E$ that satisfies Invariant~(I'), we 
identify an (I')-safe pair in $E$ using appropriate analogues of 
Lemmas~\ref{lem: m = 3, preserves invariant, n = 5}, 
\ref{lem: m = 4, preserves invariant, n = 5},
and \ref{lem: m = 5, preserves invariant, n = 5} (with the last lemma extended to cover the
case when $\nofconcentrations{E} = 6$). Intuitively, these proofs are now in fact easier,
because there are more choices for (I')-safe pairs.

Here is one possible way to adapt these proofs, essentially by reducing the argument to the case
for $n=5$. Instead of $E$, consider the configuration $\tildeE = E -\braced{\tildee}$ 
where $\tildee\in E$ is a concentration with maximum multiplicity. Then
$\nofdroplets{\tildeE}=5$, and $E$ is non-blocking (for $n=6$) if and only if $\tildeE$ is
non-blocking (for $n=5$).
We now apply the analysis from Lemmas~\ref{lem: m = 3, preserves invariant, n = 5}, 
\ref{lem: m = 4, preserves invariant, n = 5},
and \ref{lem: m = 5, preserves invariant, n = 5} to identify an (I')-safe mixing pair for $\tildeE$ 
(except that we maintain $3$-incongruence instead of $5$-incongruence,
since the only odd prime factor of $6$ is $3$). We then use this pair for $E$.

There is one case where this suggested adaptation of the proof needs a modification,
namely Case~1.1 in Lemma~\ref{lem: m = 3, preserves invariant, n = 5}.
In this case we have $E=\braced{3:e_1, 2:e_2, e_3}$, so $\tildee = e_1$ and 
$\tildeE = \braced{2:e_1, 2:e_2, e_3}$.
We also have that $e_1$ and $e_2$ are even, and
the proof uses the fact that $\half(e_1+e_2) = e_3$ implies that $\tildeE$ is near-final,
and thus $(e_1,e_2)$ is a near-final pair.
However, this implication is not true for $E$.
Nevertheless, pair $(e_1,e_2)$ is (I')-safe for $n=6$ even if $\half(e_1+e_2)= e_3$,
by applying Lemma~\ref{lem: at most one mod-unsafe pairs}(c) to show that after mixing
condition~(I.2) holds, and using the fact that the new configuration has two non-singletons $e_1$ and $e_3$
to show that condition~(I.1') holds.


\section{Polynomial bound for the number of mixing operations}
\label{sec: polynomial proof}


Let $C\ingroundset\integers$ with $\average(C)\in\integers$ and 
$\nofdroplets{C}=n \ge 4$ be a configuration that satisfies Condition~$\MixCond$.
The existence of a perfect-mixing graph for $C$ was established in 
Section~\ref{sec: sufficiency of condition mc}. This graph, however, might be very
large -- it can be shown that if arbitrary droplets are mixed at each step then
it might take an exponential number of steps for the process to converge.
In this section we prove Theorem~\ref{thm: miscibility characterization}(b),
namely that $C$ can be perfectly mixed
with precision at most $1$ and in a polynomial number of mixing operations. 
The essence of the proof is
to show that in the construction in  Section~\ref{sec: sufficiency of condition mc}
it is possible to choose a mixing operation at each step so that the overall
number of steps will be polynomial in the input size.
We assume here that the reader is familiar with the results from  Section~\ref{sec: sufficiency of condition mc};
in fact, some of the lemmas or observations from that section may be used here occasionally
without an explicit reference.

It is sufficient to show that configuration $E$, constructed from $C$ 
in Section~\ref{subsec: proof for n greater-equal than 7}, 
is perfectly mixable with precision $0$ in a polynomial number of mixing operations.
(As described in Section~\ref{subsec: proof for n greater-equal than 7},
constructing such $E$ from $C$ requires only two mixing operations. Recall that this
construction also involves a linear mapping, but this mapping does not change the
constructed mixing graph.)
For this reason we will assume in this section that $E$ is the initial configuration.

Recall that $\potential(E)=\sum_{e\in E}(e-\hatmu)^2$, where
$\hatmu =  \average(E)\in\integers$. We use $\potential(E)$ to measure the progress of
the mixing process, ultimately showing that $\potential(E)$
can be decreased down to $0$ after a number of steps that is polynomial in  
the initial value of $\log\potential(E)$, and thus also in $s(E)$, the size of $E$.
(What we actually show is that after the mixing process we achieve $\potential(E)\leq 1$. 
Since $E\cup \braced{\hatmu} \subseteq \integers$, this implies that in fact $\potential(E)=0$.)

The general idea is to always mix two concentrations whose difference 
is large enough, so that after a polynomial number of mixing operations,
$\potential(E)$ decreases at least by a factor of $\half$. This is sufficient to 
establish a polynomial bound for the whole process.  Of course, the  pair of
concentrations that we mix must either preserve the corresponding invariant
or produce a near-final configuration.
When $n$ is a power of $2$, it is relatively
simple to identify good pairs to mix (see Section~\ref{sec: polynomial - near final}), 
because then we only need to ensure that the 
concentrations mixed at each step have the same parity. This also easily
extends to near-final configurations.
For arbitrary configurations $E$ (that satisfy the appropriate invariant), however,
identifying such good pairs to mix is more challenging (see 
Sections~\ref{sec: an exponential bound} and~\ref{sec: polynomial - arbitrary n}).

In the following sub-sections, we will present some auxiliary mixing sequences 
that will be later combined to construct a polynomial-length mixing sequence for $E$.
To streamline the arguments, we will focus on estimating the length of these sequences;
that these mixing sequences can actually be computed in polynomial time will be implicit 
in their construction. We will return to the time complexity analysis
in Section~\ref{sec: polynomial running time}.


\subsection{Auxiliary observations}
\label{sec: polynomial - auxiliary observations}

Let $A\ingroundset\integers$ with $\average(A)=\mu_A\in\integers$ and 
$\nofdroplets{A}=n_A$ be an arbitrary configuration (that is, a multiset of integers).
We define $\min(A)$ and $\max(A)$ as the lowest and highest concentrations in $A$, respectively.
We also define the diameter of $A$ as $\diameter(A) = \max(A) - \min(A)$.
Recall that the size of $A$ is defined as $s(A)=\sum_{a\in A} \log(|a|+2)$. 
We denote by $A_{even}$ and $A_{odd}$ two disjoint multisets containing
all even and odd concentrations in $A$, respectively.
We use notation $\pi\in\braced{even,odd}$ for a parity value, with
$\barpi\in\braced{even,odd}$ standing for the opposite parity (that is $\barpi\neq\pi$),
and we use these for subscripts, as in  $A_\pi$ and $A_{\barpi}$.   
In the contexts when $A$ is subjected to some mixing sequence,
we will use notation $\potential_{A,0}$ as the initial value of $\potential(A)$,
before any mixing has been performed on $A$.

The following observations show that repeatedly mixing two droplets with concentrations 
that are sufficiently far apart eventually decreases $\potential(A)$ at least by a factor of $\half$.


\begin{observation}\label{obs: far-appart mix}
Let $x,y\in A$ be two different concentrations with the same parity, and let $A'$ be
obtained from $A$ by mixing $x$ and $y$. Then:
\begin{description}
	
	\item{(a)} $\potential(A') < \potential(A)$.

	\item{(b)} If $\gamma\in\posintegers$ is a constant and
	$|x-y|\geq\diameter(A)/\gamma$, then
	$\potential(A') \;\leq\; \big(1 - \frac{1}{2\gamma^2n_A} \big)\potential(A)$.
	
\end{description}

\end{observation}

\begin{proof}
Without loss of generality we can assume that $\mu_A= 0$, so that $\potential(A)=\sum_{a\in A}a^2$.
Note that we then also have $\average(A')=\mu_{A'} = 0$.
Part~(a) now follows by simple calculation:
	\begin{equation*}
			\potential(A) - \potential(A')
			\;=\; x^2+y^2-2\Big(\frac{x+y}{2}\Big)^2 
			\;=\; \frac{(x-y)^2}{2} > 0.
	\end{equation*}

\noindent
To prove~(b), we use the above calculation and inequality $\potential(A)\le n_A \diameter(A)^2$
(that follows directly from the definition of $\potential(A)$):
	\begin{equation*}
			\potential(A) - \potential(A')
			\;=\; \frac{(x-y)^2}{2} 
			\;\geq\; \frac{\diameter(A)^2}{2\gamma^2} 
			\;\geq\; \frac{\potential(A)}{2\gamma^2n_A},
	\end{equation*}

\noindent
completing the proof.
\end{proof}


\begin{observation}\label{obs: far-appart constant mix}
	Let $\gamma\in\posintegers$ be a constant.   
	Let $A'$ be obtained from $A$ by a sequence of $2\gamma^2n_A$ mixing operations,
	each involving droplets $x,y\in A$ that satisfy $|x-y|\geq\diameter(A)/\gamma$.  
	Then $\potential(A') \le \potential(A)/2$.
\end{observation}

\begin{proof}
	Applying Observation~\ref{obs: far-appart mix}(b) repeatedly, for each of the
	$2\gamma^2n_A$ mixing steps, we obtain
	\begin{equation*}
		\potential(A')
		\leq \potential(A)\Big(1 - \frac{1}{2\gamma^2n_A}\Big)^{2\gamma^2n_A}
		\leq \frac{\potential(A)}{2},
	\end{equation*}
	where the last inequality holds because $(1-1/k)^k \leq 1/2$ for $k\geq 2$.
\end{proof}


We state two more observations that will be used later in the proof, typically
without an explicit reference.

\begin{observation}\label{obs: logPotential <= 2s(A)}
	$\log{\potential(A)}\leq 2s(A)$.
\end{observation}

\begin{proof}
	Let $\psi_A(x) =  \sum_{a\in A}(a-x)^2$. By calculus, we have that
	$\psi_A(x)$ is minimized for $x = \mu_A$, and thus
	\begin{equation*}
		 \potential(A) \;=\; \psi_A(\mu_A)
						\;\le\; \psi_A(0)
						\; = \;  \sum_{a\in A}a^2
						\;\le    \Big(\,\sum_{a\in A} |a| \,\Big)^2.
	\end{equation*}   
	Therefore 
	\begin{align*}
		\log\potential(A) &\leq\; \log\Big(\,\sum_{a\in A} |a| \,\Big)^2
		\\
		&\leq\; 2\log\prod_{a\in A}\big(|a|+2\big)
		\;=\; 2\sum_{a\in A}\log(|a|+2)
		\;=\; 2s(A),
	\end{align*}	
completing the proof.
\end{proof}

\begin{observation}\label{obs: psi(A') <= psi(A) for A' subseteq A}
	Let $A'\subseteq A$. Then $\potential(A')\leq\potential(A)$.
\end{observation}      

\begin{proof}  
Using the properties of the 
function $\psi_A(x)$ from the proof of Observation~\ref{obs: logPotential <= 2s(A)},
we have
\begin{equation*}	 
	 \potential(A') \;=\; \psi_{A'}(\mu_{A'})
	 				\;\leq\; \psi_{A'}(\mu_{A})   
	 				\;\leq\; \psi_{A}(\mu_{A})  
	   				\;=\; \potential(A),
\end{equation*}	
where the second inequality follows from $A'\subseteq A$.  	
\end{proof}


\subsection{Proof for near-final configurations}
\label{sec: polynomial - near final}


In this sub-section we prove that $E$ can be perfectly mixed with precision $0$
in a polynomial number of mixing operations when $\nofdroplets{E}=n$ is a power of $2$.
This extends easily to near-final configurations.

Assume that $n$ is a power of $2$. To establish our upper bound, we
consider a sequence of mixing operations on $E$, where at each step we mix
farthest-apart same-parity concentrations in $E$. These concentrations
are both in the same set $E_\pi$, but the two droplets obtained from mixing can be
either in $E_\pi$ or in $E_{\bar{\pi}}$. Thus this mixing sequence might change
cardinalities and structure of these sets over time. To obtain our bound, we 
will need to analyze how these sets evolve during the segments of mixing operations
that preserve parity (namely when mixing in each $E_\pi$
produces a concentration in the same $E_\pi$).

More specifically, the proof idea is this: By the choice of the mixing pair,
if the mixed concentrations are far apart (say, they differ at least by $\diameter(E)/3$), 
then the value of $\potential(E)$ will significantly decrease, per Observation~\ref{obs: far-appart mix}(b).
Otherwise, if the mixed concentrations are close to each other, the subsets $E_{even}$ and
$E_{odd}$ must be separated by a gap of at least $\diameter(E)/3$. For such a 
separated configuration, the decreases of $\potential(E)$ may be very small (because
the diameters of $E_{even}$ and $E_{odd}$ may be small compared to $\diameter(E)$).
We show, however, that after at most a polynomial number of steps there will be a
mixing operation in some $E_\pi$ that will produce a concentration with
parity $\barpi$. This guarantees that at the next step the diameter of $E_{\barpi}$ is
large, so in this step $\potential(E)$ will decrease significantly.     

The above intuition is formalized in  the proof of Lemma~\ref{lem: small mixes} below.
We first prove the following auxiliary lemma.  


\begin{lemma}\label{lem: same-parity mixability}     
	Let $\pi\in\braced{even,odd}$. 
	Consider a sequence of farthest-apart mixing operations on $E_\pi\subseteq E$, such
	that each mixing in $E_\pi$ produces a concentration with the same parity $\pi$.
	If this sequence contains $4ns(E)$ mixing operations, then after
	this sequence we have $\nofconcentrations{E_\pi} = 1$.
\end{lemma} 

\begin{proof}
	Recall that $\potential_{E_\pi,0}$ denotes the value of $\potential(E_\pi)$ before any mixing has been performed.
	Let $\nofdroplets{E_\pi} = n_\pi$.
	By Observation~\ref{obs: far-appart constant mix} (with $\gamma =1$), if we
	repeatedly mix $\min(E_\pi)$ with $\max(E_\pi)$, then after at most $2n_\pi$ mixes
	the value of $\potential(E_\pi)$ decreases at least by a factor of $\half$.
	Thus, after at most $2n_\pi\log{\potential_{E_\pi,0}}$ such mixing operations we must have
	$\potential(E_\pi) = 0$; that is $\nofconcentrations{E_\pi} = 1$.
	(More precisely, we achieve $\potential(E_\pi) \le 1$. Since all concentrations in $E_\pi$
	have the same parity, by simple calculation, this implies that in fact $\potential(E_\pi) = 0$.
	It is worth to point out here that this argument is different than our earlier argument for
	$\potential(E)$, since in this case $\average(E_\pi)$ may not be integer.)	
	To complete the proof, note that
	$n_\pi\leq n$ and $\potential_{E_\pi,0}\leq\potential_{E,0}$ 
	and therefore $2n_\pi\log{\potential_{E_\pi,0}} \leq 2n\log{\potential_{E,0}} \leq 4ns(E)$.
\end{proof}


\begin{lemma} \label{lem: small mixes}
	After any sequence of at most $8ns(E)$ farthest-apart same-parity 
	mixing operations, $\potential(E)$ decreases at least by a factor of $1 - 1/18n$.
\end{lemma}   

\begin{proof} 
	The proof follows the idea outlined at the beginning of this section. 
	Let $\delta=\diameter(E)$, be the (initial) diameter of $E$.
	If the farthest-apart same-parity  pair $x,y\in E$ satisfies
	$|x-y|\geq \delta/3$, then by Observation~$\ref{obs: far-appart mix}$,
	mixing $x$ and $y$ decreases $\potential(E)$ at least by a factor of $1-1/18n$, 
	and we are done. 
	
	Otherwise, we have that any two same-parity concentrations differ by at most $\delta/3$.
	Consider $E_{even}$ and $E_{odd}$. By the case assumption, these sets are not empty.
	Assume without loss of generality that $\max(E_{odd}) > \max(E_{even})$.
	We then have that $\min(E_{odd}) > \max(E_{odd}) - \delta/3$,
	$\min(E) = \min(E_{even})$  and $\max(E_{even}) < \min(E_{even}) + \delta/3$.
	Hence, $\min(E_{odd})  \ge \max(E_{even}) + \delta/3$.
	
	Now, from Lemma~$\ref{lem: same-parity mixability}$, 
	we derive that it takes at most $4ns(E)$ farthest-apart mixing operations on $E_{even}$ 
	for either an odd concentration to be produced or for $\nofconcentrations{E_{even}} = 1$ to hold 
	(similarly for $E_{odd}$). 
	Thus, if we repeatedly mix farthest-apart same-parity droplets in $E$,
	eventually, after fewer than $8ns(E)$ such mixing operations,	
	either a mixing in $E_{even}$ will produce an odd concentration or 
	a mixing in $E_{odd}$ will produce an even concentration.
	(This is true because we cannot have both $\nofconcentrations{E_{even}} = 1$ 
	and $\nofconcentrations{E_{odd}} = 1$:
	as mentioned in the proof for Lemma~\ref{lem: mixability for near-final},
	$n$ being a power-of-$2$ guarantees the existence of two distinct concentrations with same-parity.)
	
	So, if an odd $x$ was produced from a mixing in $E_{even}$, then we can mix $x$ with $\max(E_{odd})$.
	Otherwise, an even $y$ was produced from a mixing in $E_{odd}$ and we can mix $y$ with $\min(E_{even})$.
	As both $|x-\max(E_{odd})|$ and $|y-\min(E_{even})|$ are at least $\delta/3$, 
	either mixing decreases $\potential(E)$ at least by a factor of $1-1/18n$
	(see Observation~$\ref{obs: far-appart mix}$), and thus the lemma holds.
\end{proof}         


Using Lemma~$\ref{lem: small mixes}$ we can now establish a bound on the length of mixing sequences
when $n$ is a power of $2$, and, more generally, when $E$ is near-final.

\begin{theorem} \label{them: n power of two}
	If $\nofdroplets{E}=n$ is a power of $2$, then
	$E$ can be perfectly mixed with precision $0$ 
	by a mixing sequence of length at most $288n^2s^2(E)$.
\end{theorem}

\begin{proof}
	By Lemma~\ref{lem: small mixes}, after a mixing sequence of at most $8ns(E)$ 
	same-parity farthest-apart mixing operations,
	$\potential(E)$ decreases at least by a factor of $ (1-1/18n)$.
	It follows that  after at most $18n$ such mixing sequences, $\potential(E)$ decreases
	at least by a factor of $\half$ (see the proof of Observation~\ref{obs: far-appart constant mix}).   
	Consequently, after at most $18n\log{\potential_{E,0}}$ such mixing sequences,
	$E$ becomes perfectly mixed.	
	Finally, as $\log {\potential_{E,0}} \le 2 s(E)$ and 
	each mixing sequence contains at most $8ns(E)$ mixing operations,
	we obtain that the total number of mixing operations is at most 
	$(18n\log{\potential_{E,0}})\cdot (8ns(E)) \le (36ns(E)) \cdot (8ns(E)) = 288n^2s^2(E)$.
\end{proof}


\begin{theorem}\label{thm: E near-final}
	If $E$ is near-final, then $E$ can be perfectly mixed, with precision $0$,
	by a mixing sequence of length at most $144n^3s^2(E)$.
\end{theorem}

\begin{proof}
	Assume that $E$ is near-final. Recall that $\average(E) = \hatmu$.
	By definition, $E$ can be partitioned into singletons $\hatmu$ and
	at most $n/2$ disjoint multisets $A\subseteq E$, each satisfying $\average(A) = \hatmu$
	and having cardinality that is a (non-zero) power of $2$.
	By Theorem~$\ref{them: n power of two}$ above,
	each multiset $A$ in this partition can be perfectly mixed with at most $288n^2s^2(E)$ mixing operations,
	and the theorem holds.
\end{proof}


\subsection{An exponential bound on mixing sequences}
\label{sec: an exponential bound}


In this sub-section we give a simple upper bound on the number of mixing operations
to perfectly mix $E$ with precision $0$ when $\nofdroplets{E}=n\geq 5$.
This bound is exponential in $n$ --- so it's too weak for our purpose for general $n$ --- but we will
use it only to get a polynomial bound when $n$ is at most $22$
(see Theorem~\ref{thm: perfect mixability polynomial}). 

Let $E$ satisfy Invariant~$(\lambda)$, where $\lambda = I'$ for $n=5,6$ and $\lambda = I$ for $n\geq 7$.
Also, let $A\subseteq E$.
We say that a pair of distinct concentrations in $A$ is \emph{$(\lambda)$-safe} 
if it is $(\lambda)$-safe with respect to $E$, that is, if mixing this pair
preserves Invariant~$(\lambda)$ for $E$.
If there is no $(\lambda)$-safe pair in $A$ then we say that $A$ is \emph{$(\lambda)$-mixed}.
Recall that, since $E$ satisfies Invariant~$(\lambda)$, 
according to the properties established in Sections~\ref{subsec: proof for n greater-equal than 7} and~\ref{subsec: proof for n equal 5},
it is guaranteed that in $E$ there exists a pair of concentrations that is either $(\lambda)$-safe
or near-final;
however, this pair may not be in the set $A$ under consideration. 
(In other words, $A$ being $(\lambda)$-mixed does not mean that all droplets in $A$ have the same concentration
--- it only means that we cannot mix any pairs from $A$ in our perfect-mixing sequence for $E$.)

First, in Lemma~$\ref{lem: E-mixed}$ below we show an upper bound on the number of mixing
operations to $(\lambda)$-mix a subset $A\subseteq E$. Then, we prove Theorem~\ref{thm: perfect mixability upper bound}
that gives an upper bound for the number of mixing operations to perfectly mix $E$ with precision $0$.


\begin{lemma}\label{lem: E-mixed}
	Assume that $E$ satisfies Invariant~$(\lambda)$. 
	If $A \subseteq E$ and $\nofdroplets{A} = k$, 
	then $A$ can be $(\lambda)$-mixed with precision $0$ 
	by a mixing sequence of length at most $(8k^3s(A))^k$.
\end{lemma}

\begin{proof}	
	We prove the lemma by induction with respect to $k$.
	Let $\phi(A)$ denote the number of furthest-apart $(\lambda)$-safe mixing operations needed to $(\lambda)$-mix $A$. We
	use induction to prove that $\phi(A) \le (8k^3s(A))^k$. This will imply the lemma.

	In the base case, when $k\le 1$, we have $\phi(A) = 0$, since $A$ is trivially $(\lambda)$-mixed.
	So for the rest of the proof assume that $k\ge 2$ and that every $A' \subset A$
	with $\nofdroplets{A'} = k'<k$ can be $(\lambda)$-mixed in $\phi(A')\leq(8{k'}^3s(A'))^{k'}$ mixing operations.
	We next show that $A$ can be $(\lambda)$-mixed in at most $(8k^3s(A))^k$ mixing operations.
	
	If $A$ is already $(\lambda)$-mixed then we are done.
	Otherwise, let $x,y\in A$ be the furthest-apart $(\lambda)$-safe pair. 
	If $|x-y|\geq \diameter(A)/k$, we will call pair $x,y$ \emph{acceptable}.
	If this pair $x,y$ is acceptable then mixing $x$ and $y$ decreases 
	$\potential(A)$ at least by a factor of $1-1/2k^3$ (by applying 
	Observation~$\ref{obs: far-appart mix}$ with $\gamma = k$ and $n_A = k$).	
	Then, by Observation~$\ref{obs: far-appart constant mix}$, after at most $2k^3$ 
	such acceptable mixing operations $\potential(A)$ decreases at least by a factor of $\half$, and
	it follows that after at most $2k^3\log{\potential_{A,0}}$ such mixing operations, $A$ becomes $(\lambda)$-mixed.
	However, there may be many steps where there is no acceptable pair.
	To address this, we show below a strategy (basically a divide-and-conquer approach)
	that will bound the number of consecutive steps
	needed for an acceptable pair to appear.
	
	Let $\delta=\diameter(A)$ before any mixing operation.
	Divide the interval $[\min(A), \max(A)]$ into $k$ equal segments 
	such that for at least one segment, say $[l,r]$, no concentration in $A$ lies withing
	the open interval $(l,r)$.
	Split $A$ into $A_1$ and $A_2$ such that $\max(A_1) \leq l$ and $\min(A_2) \geq r$.
	Let $k_i = |A_i|$, for $i=1,2$.
	By our inductive assumption, each $A_i$ can be $(\lambda)$-mixed in
	$\phi(A_i) \le  (8k_i^3s(A_i))^{k_i} < (8k^3s(A))^{k-1}$ mixing operations, respectively.
	Therefore, after at most $\phi(A_1) + \phi(A_2)+1 \le 2(8k^3s(A))^{k-1}$
	mixing operations, one of two things must happen:
	either (i) $A$ is already $(\lambda)$-mixed, or
	(ii) there are $x\in A_1$ and $y\in A_2$ that form an $(\lambda)$-safe pair.
	In option~(ii), $x,y$ are in fact an acceptable pair, by the choice of the segment $[l,r]$.
	As argued in the paragraph above, if we repeat the above strategy, 
	option~(ii) can be repeated no more than $2k^3\log{\potential_{A,0}}$
	times before $A$ becomes $(\lambda)$-mixed. Since $\log{\potential_{A,0}} \le 2s(A)$,
	it follows that
	\begin{align*}
			\phi(A)
			\;&\leq\; 2k^3\log\potential_{A,0} \cdot 2(8k^3s(A))^{k-1}
			\\
			&\leq\; 4k^3s(A)\cdot 2(8k^3s(A))^{k-1}
			\;\leq\; (8k^3s(A))^k,
	\end{align*}
	completing the inductive step and the proof of the lemma.
\end{proof}


\begin{theorem} \label{thm: perfect mixability upper bound}
	If $|E|=n\geq 5$, then $E$ can be perfectly mixed with a mixing sequence of length at most
	$(8n^3s(E))^n + 144n^3s^2(E)$.
\end{theorem}

\begin{proof}
	Let $E$ satisfy Invariant~$(\lambda)$, where $\lambda = I$ for $n\ge 7$ and $\lambda = I'$ for $n = 5,6$.
	We know that if $E$ satisfies Invariant~$(\lambda)$	then $E$ has a pair of concentrations
	that is either $(\lambda)$-safe or near-final. 
	Therefore, we first $(\lambda)$-mix $E$ using Lemma~$\ref{lem: E-mixed}$, 
	which produces a near-final pair.
	(Such a near-final pair becomes available because, as sown in
	Sections~\ref{subsec: proof for n greater-equal than 7}	and~\ref{subsec: proof for n equal 5},
	$E$ satisfying Invariant~($\lambda$) guarantees the existence of a pair that is either ($\lambda$)-safe
	or near-final, and as $E$ is $(\lambda)$-mixed, the later must hold.)
	We then mix this near-final pair, producing $E$ near-final 
	that can be perfectly mixed using Theorem~$\ref{thm: E near-final}$.	
	The total number of mixing operations in this sequence is at most $(8n^3s(E))^n + 144n^3s^2(E)$.
\end{proof}


\subsection{A polynomial bound on mixing sequences}
\label{sec: polynomial - arbitrary n}


We now complete the proof of Theorem~\ref{thm: miscibility characterization}(b).
In Section~\ref{sec: sufficiency of condition mc} we have already established an existence of a perfect-mixing 
sequence for $C$ with intermediate precision $1$.
It still remains to show that this mixing sequence can be modified to have length that is polynomial in $s(C)$.

As explained in the beginning of Section~\ref{sec: polynomial proof}, it remains to
give such a bound for the configuration $E$ constructed from $C$ in Section~\ref{subsec: proof for n greater-equal than 7}
and for mixing sequences with precision $0$ (that is, with only integral concentrations).
Recall that in Section~\ref{sec: polynomial - near final} we gave a polynomial bound if $E$ is near-final.
In Section~\ref{sec: an exponential bound} we gave a bound for arbitrary configurations $E$ that is
exponential only in $n$, thus yielding a polynomial bound for constant $n$.
We will be using these properties in our construction in this section. 
If $E$ is not near-final and $n$ is arbitrary, the general idea for achieving a polynomial bound
is similar to that from Sections~\ref{sec: polynomial - near final} and~\ref{sec: an exponential bound}:
we try to mix a pair of concentrations that are far apart,
to guarantee that we make a quick progress, as measured by the decrease of $\potential(E)$. 
(This is captured by Observation~\ref{obs: far-appart mix}, that will be used frequently and without
explicit reference.) Although this is not posible in each step, we (essentially)
show that for any $E$ there is a polynomial-length sequence of mix operations after which such a
far-apart pair will exist, yielding an overall polynomial bound on a perfect-mixing sequence.
Additionally, to obtain such a perfect-mixing sequence,
at each step we must mix pairs that are either $(\lambda)$-safe, 
for corresponding $\lambda\in\braced{I,I'}$, or near-final.

In this section, for better clarity, we chose to use a top-down presentation style, starting with the statement
of the main theorem (Theorem~\ref{thm: perfect mixability polynomial} below) and its proof, while
all the necessary lemmas that cover different cases from that proof are presented later.
Also, throughout this section, $\gamma$ denotes a small integral constant; for
concreteness we can assume that $\gamma = 2$.


\begin{theorem} \label{thm: perfect mixability polynomial}
	$E$ can be perfectly mixed with precision $0$ 
	in a polynomial number of mixing operations.
\end{theorem} 

\begin{proof}
	If $n = 4$ (and, more generally, if $n$ is a power of $2$), we can simply use 
	Theorem~\ref{them: n power of two} and we are done.
	Thus, we can assume that $n\geq 5$.	
	If $n<22$, we can perfectly mix $E$ using Theorem~\ref{thm: perfect mixability upper bound},
	where the bound on the length of the mixing sequence is polynomial if $n$ is constant.
	Thus in the rest of the argument below we assume that $n \geq 22$.
	
	We show that as long as $E$ satisfies Invariant~(I) (which is the invariant that applies when $n\ge 22$),
	we can find a polynomial-length sequence of mixing operations that will
	decrease $\potential(E)$ by a factor of $1-\Omega(1/n)$. This will be sufficient to prove the theorem.
	
	So assume that $E$ is not yet perfectly mixed and let $\delta=\diameter(E)$. Let $\pi$ be the parity
	for which $\nofdroplets{E_{\pi}}\geq\nofdroplets{E_{\barpi}}$.
	(Recall that we say that $A\subseteq E$ is (I)-mixed if there is no (I)-safe pair in $A$,
	that is, there no pair of distinct concentrations in $A$ whose mixing preserves Invariant~(I) for $E$.
	We will often use this terminology when mixing pairs of concentrations in $E_\pi,E_\barpi\subseteq E$.)
	We consider several cases.
	
	
	\medskip\noindent
	\mycase{1} $\nofconcentrations{E_{\pi}}=1$.
	Let $a\in E_{\pi}$ (that is, $a$ is the only concentration in $E_{\pi}$). 
	Observe that $E_{\barpi}\neq\emptyset$. We have two sub-cases:
	
	\begin{description}

		\item\mycase{1.1} $\min(E_{\bar{\pi}}) < a < \max(E_{\bar{\pi}})$.
		In this case, per Lemma~\ref{lem: farm mix in E odd overlapping E even},
		there are at most two consecutive mixing operations 
		(involving either (I)-safe or near-final pairs)
		after which $\potential(E)$ decreases at least by a factor of $1 - 1/32n$.

		\item \mycase{1.2} $a < \min(E_{\bar{\pi}})$ or $a > \max(E_{\bar{\pi}})$. In this case,
		per Lemma~\ref{lem: far mix in E odd non-overlapping E even}, there is a sequence of
		at most $128ns(E)+1$ (I)-safe mixing operations after which
		$\potential(E)$ decreases at least by a factor of $1 - 1/32\gamma^2n$.
		
	\end{description}	

	
	\medskip\noindent
	\mycase{2} $\nofconcentrations{E_{\pi}}\geq 2$ and $\diameter(E_\pi)\geq\delta/\gamma$.
	In this case, per Lemma~\ref{lem: far mix in E even},	
	there is an (I)-safe pair in $E_{\pi}$ whose mixing decreases
	$\potential(E)$ at least by a factor of $1 - 1/8\gamma^2n$.
	
	
	\medskip\noindent
	\mycase{3} $\nofconcentrations{E_{\pi}}\geq 2$ and $\diameter(E_\pi)<\delta/\gamma$.
	Using Lemma~\ref{lem: E-mix c_pi n >= 22}, there is a sequence of at most
	$128ns(E)$ (I)-safe mixing operations in $E_\pi$, such that after this sequence
	$E_{\pi}$ is (I)-mixed. Let $E'$ be the configuration obtained from $E$
	after applying this sequence. We consider two sub-cases:	
	
	\begin{description}

		\item\mycase{3.1} $\nofdroplets{E'_{\pi}}\geq\nofdroplets{E'_{\bar{\pi}}}$.
		By Observation~\ref{obs: number of same-parity unsafe pairs} below and
		since $E'_{\pi}$ is (I)-mixed, we have that $\nofconcentrations{E'_{\pi}}=1$.		
		Hence, as in Case~$1$ above, there is a sequence of at most $128ns(E')+1$ mixing operations 
		(involving either (I)-safe or near-final pairs)
		after which $\potential(E')$ decreases at least by a factor of $1-1/32\gamma^2n$.
		Thus, $\potential(E)$ also decreases at least by a factor of $1 - 1/32\gamma^2n$.

		\item\mycase{3.2} $\nofdroplets{E'_{\pi}}<\nofdroplets{E'_{\bar{\pi}}}$.
		As $\gamma=2$, $\diameter(E'_\pi)<\delta/\gamma$, and 
		because	the mixing operations on $E_{\pi}$ produced at least one droplet 
		with parity $\bar{\pi}$, we have that 
		$\diameter(E'_{\bar{\pi}})>\delta/\gamma$.
		Hence, as in Case~$2$ above, 
		there is an (I)-safe pair whose mixing decreases
		$\potential(E')$ at least by a factor of $1 - 1/8\gamma^2n$.
		Thus, such mixing also decreases $\potential(E)$ at least by a factor of $1 - 1/8\gamma^2n$.				
	\end{description}
	
	Applying a sequence of mixing operations specified by the cases above
	results in a decrease of $\potential(E)$ at least by a factor of $1 - 1/32\gamma^2n$.
	Thus, by a simple extension to Observation~\ref{obs: far-appart constant mix},
	we obtain that after at most $32\gamma^2n$ such mixing sequences $\potential(E)$ decreases at least by half.
	It follows that after at most $32\gamma^2n\log\potential(E)$ of these mixing sequences (where $\potential(E)$
	denotes the initial potential value), $E$ becomes (I)-mixed, and,
	as $E$ satisfies Invariant~(I), there is a near-final pair in $E$ that we then mix to make $E$ near-final.
	Since $\log\potential(E)\le 2s(E)$, the length of this sequence is at most $64\gamma^2ns(E)$ (where $s(E)$ 
	represents the original size of $E$).
	Each such mixing sequence involves at most $256ns(E)+1$ mixing operations and,
	by~Theorem~\ref{thm: E near-final}, if $E$ is near-final then it can be perfectly mixed 
	by a sequence of at most $144n^3s^2(E)$	mixing operations. 
	Therefore the total number of mixing operations to perfectly mix $E$ is
	at most $2^{14}\gamma^2n^2s^2(E) + 64\gamma^2ns(E) + 144n^3s^2(E)$.
\end{proof}


\begin{observation}\label{obs: number of same-parity unsafe pairs}
	Assume that $E$ with $\nofdroplets{E}=n\geq 22$ satisfies Invariant~(I)
	and let $\pi$ be the parity for which $\nofdroplets{E_{\pi}}\geq\nofdroplets{E_{\bar{\pi}}}$.
	There is at least one droplet in $E_{\pi}$ such that, when paired with any other
	droplet in $E_{\pi}$, the resulting pair is (I)-safe.
\end{observation}

\begin{proof}
	Lemma~$\ref{lem: at most one mod-unsafe pairs}$(a) and the number of distinct odd 
	prime factors of $n$ being less than $\log_3{n}$ imply that $E$ has at most
	$2\lfloor\log_3{n}\rfloor$ droplets that are $\barp$-unsafe when paired with other droplets in $E$.
	This, and Observation~$\ref{obs: decrease number of non-singletons}$ below
	give that the number of droplets that, when mixed with other droplets in $E$
	violate Invariant~(I), is at most 
	$2\lfloor\log_3{n}\rfloor+6< n/2\leq\nofdroplets{E_{\pi}}$, since $n\geq22$.
	Thus, the lemma holds.
\end{proof}

\begin{observation}\label{obs: decrease number of non-singletons}
Assume that $E$ with $\nofdroplets{E}=n\geq 7$ satisfies Invariant~(I).
The number of droplets involved in mixing operations that decrease the number of
non-singletons in $E$ down to one is at most $6$.
\end{observation}

\begin{proof}
	First of all, if the number of non-singletons in $E$ is more than three, 
	then no mixing decreases the number of non-singletons down to one;
	similarly when mixing a non-singleton with frequency higher than two.
	Additionally, mixing two singletons does not decrease the number of non-singletons.	
	Now, $E$ satisfying Invariant~(I) implies that there are 
	at least two non-singletons $a,b\in E$.	
	We consider two types of situations where a mixing decreases the number 
	of non-singletons in $E$ down to one:
	
	\medskip\noindent\mycase{1} Mixing two non-singletons, say $a$ and $b$.
	This can happen when the frequency of both $a$ and $b$ is two each,
	leading to a total of $4$ droplets involved.
	(There could be another non-singleton $e=\half(a+b)$ in $E$, however no mixing 
	involving $e$ would decrease the number of non-singletons down to one
	because either $a$ or $b$ would remain non-singleton after the mixing.)	
	
	\medskip\noindent\mycase{2} Mixing a non-singleton with a singleton.
	(This can happen when $E$ has exactly two non-singletons, $a$ and $b$ respectively.)
	Let $a$ and $c$ be the non-singleton and singleton, respectively.	
	If $a$ is a doubleton and $b=\half(a+c)$, then mixing $a$ and $c$ decreases
	the number of non-singletons down to one.
	Similarly, if $b$ is a doubleton and $a=\half(b+d)$, for some singleton $d\in E$,
	then mixing $b$ and $d$ decreases the number of non-singletons down to one.
	Therefore, the total number of droplets (excluding $a$ and $b$, which were already
	counted in Case~$1$ above) is at most $2$.
	
	\medskip
	Therefore, the number of droplets involved in mixing operations that decrease 
	the number of non-singletons down to one is at most $6$.
\end{proof}


\begin{lemma}\label{lem: farm mix in E odd overlapping E even}
	Assume that $E$ with $\nofdroplets{E}=n\geq 7$ satisfies Invariant~(I).
	Also, assume that $\nofdroplets{E_{\pi}}\geq\nofdroplets{E_{\bar{\pi}}}\geq 1$ and
	that $\nofconcentrations{E_{\pi}} = 1$, with $a\in E_{\pi}$.
	If $\min(E_{\bar{\pi}})<a<\max(E_{\bar{\pi}})$ then there exist at most
	two consecutive mixing operations (involving either (I)-safe or near-final pairs) that
	decrease $\potential(E)$ at least by a factor of $1-1/32n$.
\end{lemma}   

\begin{proof}
	Let $\delta=\diameter(E)$ before any mixing operation.
	Assume without loss of generality that $\pi=even$.
	Since $E$ satisfies Invariant~(I), there is at least one non-singleton $b\in E_{odd}$.
	Either $\min(E_{odd})$ or $\max(E_{odd})$ is furthest from $b$, so
	assume without loss of generality that $\min(E_{odd})$ is furthest from $b$; 
	$|b-\min(E_{odd})|\geq\delta/2$.	
	If pair $(b,\min(E_{odd}))$ is (I)-safe, then we are done.
	So, assume otherwise, namely that pair $(b,\min(E_{odd}))$ is not (I)-safe;
	$a=\half(b+\min(E_{odd}))$ and $b$ (which is a doubleton) are non-singletons in $E$.
	(There could be at most one more non-singleton in $E$, namely $\min(E_{odd})$ strictly doubleton;
	otherwise $(b,\min(E_{odd}))$ would be in fact (I)-safe.)
	Consider the following cases:
	
	\medskip\noindent\mycase{1} $b = \max(E_{odd})$.
	We further consider the following two sub-cases based on the number
	of distinct concentrations in $E$:
	\begin{description}
		\item\mycase{1.1} $\nofconcentrations{E}=3$.
		Then $E=\braced{\nofdroplets{E_{odd}}-2:\min(E_{odd}),\nofdroplets{E_{even}}:a,2:b}$,
		with $\nofdroplets{E_{odd}}\leq 4$.
		If $\nofdroplets{E_{odd}}=3$, then $\min(E_{odd})$ is a singleton, and, 
		after the mixing, $\nofconcentrations{E}=2$ holds,
		which implies that $E$ is now near-final, by Lemma~\ref{lem: mixability m = 2}
		and thus $(b,\min(E_{odd}))$ is actually a near-final pair.		
		Otherwise, $\nofdroplets{E_{odd}}=4$ and $\min(E_{odd})$ is a doubleton.
		This trivially implies that after (and before) the mixing $E$ is near-final,
		so $(b,\min(E_{odd}))$ is a near-final pair.
		(Recall that $|b-\min(E_{odd})|\geq\delta/2$, so the lemma holds
		by Observation~\ref{obs: far-appart mix}.)

		\item\mycase{1.2} $\nofconcentrations{E}\geq 4$.
		This implies that there is $c\in E_{odd}$ such that $\min(E_{odd})<c<b$;
		mixing $b$ and $c$ produces a non-singleton other than $a$,
		which preserves Invariant~(I), so $(b,c)$ is (I)-safe.
		We analyze the following sub-cases:
		\begin{description}
			\vspace{-0.03in}
			\item\mycase{1.2.1} $c<a$.
			Since $|b-c|\geq\delta/2$, mixing $b$ and $c$ decreases $\potential(E)$ 
			at least by a factor of $1-1/8n$, so we mix them.
			
			\item\mycase{1.2.2} $c>a$.
			Let $d=\half(b+c)$ be the output of the mixing between $b$ and $c$.
			If $d$ is odd, then mixing $d$ and $\min(E_{odd})$ decreases $\potential(E)$ 
			at least by a factor of $1-1/8n$. (Pair $(d, \min(E_{odd}))$ is (I)-safe
			because its mixing produces a non-singleton other than $a$.)
			Otherwise, $d$ is even and, as $|d-a|\geq\delta/4$, 
			mixing $a$ and $d$ decreases $\potential(E)$ at least by a factor of $1-1/32n$.
			(Note that $(a,d)$ is (I)-safe because, after its mixing, $a$ remains non-singleton.)
		\end{description}
		
	\end{description}

	\noindent\mycase{2} $b < \max(E_{odd})$.
	This is similar to Case~$1.2.2$ above under the assumption that $b < \max(E_{odd})$
	and using $c=\max(E_{odd})$.
\end{proof}


\begin{lemma}\label{lem: far mix in E odd non-overlapping E even}
	Assume that $E$ with $\nofdroplets{E}=n\geq 7$ satisfies Invariant~(I).
	Also, assume that $\nofdroplets{E_{\pi}}\geq\nofdroplets{E_{\bar{\pi}}}\geq1$ and
	that $\nofconcentrations{E_\pi} = 1$ with $a\in E_\pi$.
	If either $a < \min(E_{\bar{\pi}})$ or $a>\max(E_{\bar{\pi}})$,	then there exists
	a mixing sequence (of (I)-safe pairs) of length at most $128ns(E) + 1$
	that decreases $\potential(E)$ at least by a factor of $1-1/32\gamma^2n$, 
	for constant $\gamma\in\posintegers$.
\end{lemma}

\begin{proof}
	Let $\delta=\diameter(E)$ before any mixing operation.
	Assume without loss of generality that $\pi=even$ and that $a<\min(E_{odd})$.
	We analyze two sub-cases:
	
	\medskip\noindent\mycase{1} $\diameter(E_{odd})\geq\delta/2\gamma$.
	$E$ satisfying Invariant~(I) implies that there is a non-singleton $b\in E_{odd}$.
	Either $\min(E_{odd})$ or $\max(E_{odd})$ is furthest from $b$, so
	assume without loss of generality that $\min(E_{odd})$ is furthest from $b$.
	As $a<\half(b+\min(E_{odd}))$, pair $(b,\min(E_{odd}))$ is (I)-safe 
	(mixing $b$ and $\min(E_{odd})$	produces a non-singleton other than $a$)
	that satisfies $|b-\min(E_{odd})|\geq\delta/4\gamma$.
	Therefore, mixing $b$ and $\min(E_{odd})$ decreases $\potential(E)$ 
	at least by a factor of $1-1/32\gamma^2n$.
	
	\medskip\noindent\mycase{2} $\diameter(E_{odd})<\delta/2\gamma$.
	This implies that $|\min(E_{odd})-a|\geq\delta/2\gamma$.
	(I)-mix $E_{odd}$ using Lemma~$\ref{lem: E-mix c_pi n >= 22}$;
	$E$ satisfying Invariant~(I) (and our proof in 
	Section~\ref{subsec: proof for n greater-equal than 7}) 
	implies that an even concentration $b$ is eventually produced.
	Since $a$'s frequency is at least $\ceil{n/2}\geq 3$, pair $(a,b)$ is (I)-safe;
	$a$ remains non-singleton after the mixing.
	Additionally, since $b > \min(E_{odd})$, $|b-a|>\delta/2\gamma$ and thus
	mixing $a$ and $b$ decreases $\potential(E)$ at least by a factor of $1-1/8\gamma^2n$.
	Finally, Lemma~$\ref{lem: E-mix c_pi n >= 22}$ takes at most 
	$128ns(E)$ mixing operations, so the lemma holds.
\end{proof}


\begin{lemma}\label{lem: far mix in E even}
	Assume that $E$ with $\nofdroplets{E}=n\geq 22$ satisfies Invariant~$(I)$.
	Also, assume that $\nofdroplets{E_{\pi}}\geq\nofdroplets{E_{\bar{\pi}}}$
	and that $\nofconcentrations{E_{\pi}}\geq 2$.
	If $\diameter(E_\pi) \geq \diameter(E)/\gamma$, for constant $\gamma\in\posintegers$, 
	then there exists an (I)-safe pair whose mixing decreases $\potential(E)$ 
	at least by a factor of $1-1/8\gamma^2n$.
\end{lemma}

\begin{proof}
	Assume without loss of generality that $\pi=even$ and let $a=\min(E_{even})$ and $b=\max(E_{even})$.
	Mixing $a$ and $b$ decreases $\potential(E)$ at least by a factor of $1-1/2\gamma^2n$.
	So, if $(a,b)$ is (I)-safe then we are done.
	Instead, assume otherwise; namely, that $(a,b)$ is not (I)-safe.
	
	The above assumption, and Observation~$\ref{obs: number of same-parity unsafe pairs}$,
	imply that there is $c\in E_{even}$ with $c\notin\braced{a,b}$ 
	for which any pair involving $c$ is (I)-safe.
	Now, either $a$ or $b$ is furthest from $c$, so	assume without loss of generality 
	that $a$ is furthest from $c$.
	Then, $|c-a|\geq\diameter(E)/2\gamma$ and thus mixing $a$ and $c$ decreases 
	$\potential(E)$ at least by a factor of $1-1/8\gamma^2n$.
\end{proof}


\begin{lemma}\label{lem: E-mix c_pi n >= 22}
	Assume that $E$ with $\nofdroplets{E}=n\geq 7$ satisfies Invariant~$(I)$.
	Then, there exists a mixing sequence (of (I)-safe pairs) 
	of length at most $128ns(E)$ that (I)-mixes $E_\pi$.
\end{lemma}

\begin{proof}
	Assume without loss of generality that $\pi = even$ and let $\nofdroplets{E_{even}}=n'$ and
	$\delta=\diameter(E_{even})$ before any mixing operation.
	If a mixing in $E_{even}$ produces droplets with an odd concentration,
	then these droplets are excluded from $E_{even}$ and included in $E_{odd}$;
	this can happen at most $n'/2$ times before $E_{even}$ becomes (I)-mixed.
	
	Assume that the number of (I)-safe pairs in $E_{even}$ is non-zero and 
	let $a=\min(E_{even})$ and $b=\max(E_{even})$.
	If an (I)-safe pair $x,y\in E_{even}$ satisfies $|x-y|\geq \delta/4$,
	then mixing $x$ and $y$ decreases 
	$\potential(E_{even})$ at least by a factor of $1-1/32n'$.
	It follows from Observation~$\ref{obs: far-appart constant mix}$
	that after at most $32n'$ such mixing operations, $\potential(E_{even})$
	decreases at least by a factor of $\half$.
	Hence, $E_{even}$ can be (I)-mixed after at most $32n'\log{\potential(E_{even})}$
	such mixing operations.
	
	We next show that if $E_{even}$ has not been (I)-mixed, then
	after at most two consecutive (I)-safe mixing operations
	either an odd concentration is produced 
	or $\potential(E_{even})$ decreases at least by a factor of $1-1/32n'$.
	Consequently, after at most $128ns(E)$ (I)-safe mixing operations, $E_{even}$ becomes (I)-mixed.
	
	So, if $(a,b)$ is (I)-safe then we are done; $|a-b|= \delta$
	and thus mixing $a$ and $b$ decreases $\potential(E)$ at least by a factor of $1-1/2n$.
	Instead, assume that $(a,b)$ is not (I)-safe and consider the following cases:
	
	\medskip\noindent\mycase{1} $E_{even}$ has only singletons.
	This implies that no mixing in $E_{even}$ decreases the number of non-singletons in $E$.
	So, let $c,d\in E_{even}$ be the furthest-apart (I)-safe pair
	and $x=\half(c+d)$ be the output of their mixing.
	If either $|c-d|\geq \delta/2$ or $x$ odd holds, then we are done.
	Otherwise, let $y\in\braced{a,b}$ be furthest from $x$.
	By the choice of $y$, we have that $|x-y|\geq \delta/2$,
	and as $x$ is a doubleton, $(x,y)$ is (I)-safe and thus we mix it.
	
	\medskip\noindent\mycase{2} $E_{even}$ has exactly one non-singleton $c$.
	Either $a$ or $b$ is furthest from $c$, so assume without loss of generality 
	that $b$ is furthest from $c$; $|b-c|\geq \delta/2$.
	If $(b,c)$ is (I)-safe then we mix it and we are done.
	Otherwise, $(b,c)$ is not (I)-safe, therefore there are exactly two non-singletons in $E$,
	$c\in E_{even}$ (which is a doubleton) and $e=(b+c)/2\notin E_{even}$.
	Since there is an (I)-safe pair in $E_{even}$, there is some $d\in E_{even}$
	such that $d\notin\braced{b,c}$. Then, $(c,d)$ is (I)-safe because 
	mixing $c$ and $d$ produces a non-singleton $x=\half(c+d)$ other than $e$.
	So, mix $c$ and $d$. If $|c-d|\geq\delta/2$ then we are done; similarly if $x$ is odd.
	Otherwise, $(b,x)$ is (I)-safe ($e\neq\half(b+x)$) and $|x-b|\geq\delta/4$, so we mix it.
	
	\medskip\noindent\mycase{3} $E_{even}$ has at least two non-singletons $c<d$.
	We analyze two sub-cases:
	\begin{description}
		\item{\mycase{3.1}} $a=c$.
		Since $(a,b)$ is not (I)-safe, then mixing $b$ and $c$ decreases
		the number of non-singletons down to one. This gives us two sub-cases:
		\begin{description}
			\item{\mycase{3.1.1}} $d=b$.
			$E_{even}$ having an (I)-safe pair implies that there is some $x\in E_{even}$ such that
			$x\notin\braced{c,d}$. Either $c$ or $d$ is furthest from $x$, so			
			assume without loss of generality that $d$ is furthest from $x$; $|d-x|\geq\delta/2$.
			Mixing $x$ and $d$ produces a non-singleton other than $c$, so $(x,d)$ is (I)-safe
			and we mix it.			

			\item{\mycase{3.1.2}} $d<b$.
			This implies that $d=\half(c+b)$, $b-d=\delta/2$ and $(b,d)$ is (I)-safe
			(mixing $b$ and $d$ produces a non-singleton other than $c$), so we mix $b$ and $d$.
		\end{description}

		\item{\mycase{3.2}} $a<c$.
		If $|a-c|\geq\delta/2$ then we mix $a$ and $c$ and we are done;
		$\half(a+c) < d$ implies that mixing $a$ and $c$ produces a non-singleton other than $d$,
		so $(a,c)$ is (I)-safe.
		Instead, let $|a-c|<\delta/2$ which implies $|b-c|\geq\delta/2$.
		Now, if pair $(b,c)$ is (I)-safe then we mix it and we are done, 
		so also assume that $(b,c)$ is not (I)-safe.
		(We cannot have $b=d$ because that would contradict $(a,b)$ not being (I)-safe;
		$|a-c|<\delta/2$ implies $c < \half(a+b)$, so mixing $a$ and $b$ would produce
		a non-singleton other than $c$.)
		Therefore, as $(b,c)$ is not (I)-safe, we have that $d = \half(b+c)$ and mixing $b$ and $d$ 
		produces a non-singleton other than $c$. Hence, $(b,d)$ is (I)-safe
		satisfying $|b-d|\geq\delta/4$, and we mix it.
	\end{description}

	All the pairs mentioned above have difference at least $\delta/4$, so their mixing decreases 
	$\potential(E_{even})$	at least by a factor of $1-1/32n'$.

\end{proof}


\section{Polynomial running time}
\label{sec: polynomial running time}



	
	In this section we address Theorem~\ref{thm: miscibility characterization}(c), namely
	the claim that in polynomial time we can test 
	whether a given configuration $C$ is perfectly mixable and, if so, we can
	compute a polynomial-size perfect-mixing graph for $C$ --- also in polynomial time.
	
	Let $C$ be the input configuration. To test whether $C$ is perfectly mixable we just
	need to test whether it satisfies Condition~{\MixCond}.
	That this can be done in polynomial
	time follows directly from Corollary~\ref{cor: mc prime power factors} 
	in Section~\ref{sec: some auxiliary lemmas}.
	To justify this, recall that
	the input size is $s(C) = \sum_{c\in C} \log(|c|+2)\ge n$.
	Thus the factoring of $n$ can be computed in time polynomial in the input size.
	As $n$ has at most $\log n$ distinct odd prime factors and
	each such prime factor has at most $\log{c_{max}}$ powers that are no bigger than $c_{max}$,
	the total number of $b$'s that need to be considered 
	is at most $\log{n}\log{c_{max}}$, which is polynomial in $s(C)$.
	Putting it all together (see the pseudo-code in Algorithm~\ref{alg:perfect mixability}), 
	we obtain that testing perfect-mixability can be indeed accomplished in polynomial time, 
	thus proving the first part of Theorem~\ref{thm: miscibility characterization}(c).

	\begin{algorithm}[H]
		\caption{PerfectMixabilityTesting($C$)}
		\label{alg:perfect mixability}
		\begin{algorithmic}[1]
			\State{$n\assign\nofdroplets{C}$}
			\State{$\mu\assign \average(C)$}
			\State{$c_{max}\assign$ maximum absolute concentration in $C$}
			\State{$P\assign$ powers of odd prime factors of $n$ that are at most $c_{max}$}
			\ForAll{$p\in P$}
			\If{$C$ is $p$-congruent but $C\cup\braced{\mu}$ is not}
			\State{return false}
			\EndIf
			\EndFor
			\State{return true}
		\end{algorithmic}
	\end{algorithm}

	Now, assume that $C$ is perfectly mixable.
	To prove the second part of Theorem~\ref{thm: miscibility characterization}(c),
	namely that a polynomial-size perfect-mixing graph for $C$ can be constructed in polynomial time,
	the idea is to follow the construction from the proof in Section~\ref{sec: polynomial proof} 
	for Theorem~\ref{thm: miscibility characterization}(b).
	In this proof, at each step we choose a mixing pair in $E_\pi\subseteq E$, for $\pi\in\braced{even,odd}$,
	that is $(\lambda)$-safe (for the appropriate invariant $\lambda\in\braced{I,I'}$), 
	and sufficiently far apart (with respect to either 
	the entire configuration or a subset of the same), until $E$ becomes near-final.
	The former conditions can be checked in polynomial time 
	and there are only quadratically many pairs to try. 
	
	The only remaining obstacle is that, the way the argument is presented above, we would
	have verify at each step whether $E$ is near-final and, if it is, to
	find its near-final partition.  This is not possible in polynomial time (see the remarks below).
	We now explain how to circumvent this problem.
	For $n<22$, then this is simple, since $n$ is a constant, so we focus on the case when $n\ge 22$.
    For our construction in Section~\ref{sec: polynomial - arbitrary n},
	we can just continuously mix $(\lambda)$-safe pairs until $E$ becomes $(\lambda)$-mixed.
	When this happens, as $E$ satisfies Invariant~($\lambda$),
	it holds that there is a near-final pair in $E$.
	Furthermore, as explained in Case~1.1 of the proof of Lemma~\ref{lem: farm mix in E odd overlapping E even},
	this $E$ has the form $E=\braced{f_1:c_1,f_2:c_2,f_3:c_3}$ 
	with $f_1\leq 2$, $f_2=2$ and $\half(c_1+c_2)=c_3$,
	so $(c_1,c_2)$ is a near-final pair.
	In this case, computing the corresponding ``near-final''
	partition of $E$ in polynomial-time is easy:
	if $f_1=1$, then mixing $(c_1,c_2)$ produces $E$ satisfying $\nofconcentrations{E}=2$
	and the near-final partition is given in Lemma~\ref{lem: mixability m = 2}.
	Else $f_1=2$ and the near-final partition involves sets of the form 
	$\braced{c_1,c_2}$ and $\braced{c_3}$, where $c_3=\average(E)$.


\emph{Remarks.}
The last two paragraphs bring up a question of whether checking the near-final property can be
done in polynomial time. While this is not essential to our algorithm, we address it here
for the sake of interested readers.

Let $\nofdroplets{E} = n$. Observe first that if $E$ is near-final then it has a
near-final partition where all sets in the partition have different cardinalities.
The reason is simple: if two sets have the same cardinality, we can merge them into
one set, retaining the average value and the property that all cardinalities are powers of $2$.
As a consequence of this, the cardinalities of the sets in $E$'s partition are
uniquely determined by $n$, namely these are the powers of $2$ that appear in the binary
representation of $n$. Let $\omega = \mysum(E) = n\cdot\hatmu$
be the sum of all concentrations in $E$, and let
$k$ be some cardinality (a power of $2$) in this partition of $E$. 
A subset $F\subseteq E$ of cardinality $k$ has $\average(F) = \hatmu$ if and only if
$\mysum(F) = \omega k/n$. 

The above properties lead to an $\NP$-completeness proof.
The idea is this. Consider a restricted variant of the problem where
$n = 3\cdot 2^l$, for some $l$. Using the notation from the paragraph above, we need to
determine whether $E$ has a subset $F$ of cardinality $k = n/3$ and
total sum $\mysum(F) = \omega/3$. This is a variant of the $\Partition$
problem that can easily be shown to be $\NP$-complete. (Details are left to the reader.)


\section{Final Comments}
\label{sec: final comments}

In this paper we gave a complete characterization of perfectly mixable sets,
as well as a polynomial-time algorithm that tests perfect mixability and,
for perfectly mixable sets, computes a polynomial-size perfect-mixing graph.   

The computational complexity of {\MixReachability} remains wide open, even for
the special case of {\MixProducibility}, where
the inputs consist of pure reactant and buffer droplets.
The only hardness result related to {\MixReachability} that we can prove is
that its modified variant, where we ask whether $T$ is reachable from $I$ via
a graph with a fixed (constant) depth is $\NP$-hard (see Appendix~\ref{sec: NP-hardness}).

There are a number of other open questions about computing mixing graphs. For the
case when the target consists of just 
one droplet~\cite{thies2008abstraction,roy2010optimization,huang2012reactant,chiang2013graph},
it is not known whether minimizing waste can be done efficiently, nor even whether
the minimum-waste value is computable.
Unsurprisingly, no complexity results are known for other objective functions, like
minimizing reactant usage or minimizing the number of micro-mixers.

As another example, one can consider the problem of computing a mixing graph
that realizes a given linear mapping. More specifically,
with each mixing graph $G$ we can associate a linear mapping $L_G$ that
maps the vector of concentration values on its inputs to the vector of
concentration values on its outputs. In this problem,
given a linear mapping $L$, we ask whether there exists a mixing graph $G$ with $L_G = L$. 
An algorithm for this problem would allow us to construct mixing graphs that 
can produce a collection of target sets, by changing the buffer/reactant combinations on input.



\begin{thebibliography}{10}

\bibitem{chiang2013graph}
Ting-Wei Chiang, Chia-Hung Liu, and Juinn-Dar Huang.
\newblock Graph-based optimal reactant minimization for sample preparation on
  digital microfluidic biochips.
\newblock In {\em 2013 International Symposium on VLSI Design, Automation and
  Test (VLSI-DAT)}, pages 1--4. IEEE, 2013.

\bibitem{dinh2014network}
Trung~Anh Dinh, Shinji Yamashita, and Tsung-Yi Ho.
\newblock A network-flow-based optimal sample preparation algorithm for digital
  microfluidic biochips.
\newblock In {\em 19th Asia and South Pacific Design Automation Conference
  (ASP-DAC)}, pages 225--230. IEEE, 2014.

\bibitem{einav2008discovery}
Shirit Einav, Doron Gerber, Paul~D Bryson, Ella~H Sklan, Menashe Elazar,
  Sebastian~J Maerkl, Jeffrey~S Glenn, and Stephen~R Quake.
\newblock Discovery of a hepatitis {C} target and its pharmacological
  inhibitors by microfluidic affinity analysis.
\newblock {\em Nature Biotechnology}, 26(9):1019--1027, 2008.

\bibitem{hayden2014automated}
Erika~Check Hayden.
\newblock The automated lab.
\newblock {\em Nature News}, 516(7529):131, 2014.

\bibitem{hsieh1998automated}
Frank Hsieh, Hasmik Keshishian, and Craig Muir.
\newblock Automated high throughput multiple target screening of molecular
  libraries by microfluidic {MALDI-TOF MS}.
\newblock {\em Journal of Biomolecular Screening}, 3(3):189--198, 1998.

\bibitem{hsieh2012reagent}
Yi-Ling Hsieh, Tsung-Yi Ho, and Krishnendu Chakrabarty.
\newblock A reagent-saving mixing algorithm for preparing multiple-target
  biochemical samples using digital microfluidics.
\newblock {\em IEEE Transactions on Computer-Aided Design of Integrated
  Circuits and Systems}, 31(11):1656--1669, 2012.

\bibitem{huang2012reactant}
Juinn-Dar Huang, Chia-Hung Liu, and Ting-Wei Chiang.
\newblock Reactant minimization during sample preparation on digital
  microfluidic biochips using skewed mixing trees.
\newblock In {\em Proceedings of the International Conference on Computer-Aided
  Design}, pages 377--383. ACM, 2012.

\bibitem{huang2013reactant}
Juinn-Dar Huang, Chia-Hung Liu, and Huei-Shan Lin.
\newblock Reactant and waste minimization in multitarget sample preparation on
  digital microfluidic biochips.
\newblock {\em IEEE Transactions on Computer-Aided Design of Integrated
  Circuits and Systems}, 32(10):1484--1494, 2013.

\bibitem{kim2008serial}
Choong Kim, Kangsun Lee, Jong~Hyun Kim, Kyeong~Sik Shin, Kyu-Jung Lee, Tae~Song
  Kim, and Ji~Yoon Kang.
\newblock A serial dilution microfluidic device using a ladder network
  generating logarithmic or linear concentrations.
\newblock {\em Lab on a Chip}, 8(3):473--479, 2008.

\bibitem{li2013probing}
Peng Li, Zackary~S Stratton, Ming Dao, Jerome Ritz, and Tony~Jun Huang.
\newblock Probing circulating tumor cells in microfluidics.
\newblock {\em Lab on a Chip}, 13(4):602--609, 2013.

\bibitem{marle2005microfluidic}
Leanne Marle and Gillian~M Greenway.
\newblock Microfluidic devices for environmental monitoring.
\newblock {\em TrAC Trends in Analytical Chemistry}, 24(9):795--802, 2005.

\bibitem{meyers2009integer}
Carol~A. Meyers and Andreas~S. Schulz.
\newblock Integer equal flows.
\newblock {\em Operations Research Letters}, 37(4):245--249, 2009.

\bibitem{mitra2012chip}
Debasis Mitra, Sandip Roy, Krishnendu Chakrabarty, and Bhargab~B Bhattacharya.
\newblock On-chip sample preparation with multiple dilutions using digital
  microfluidics.
\newblock In {\em IEEE Computer Society Annual Symposium on VLSI (ISVLSI)},
  pages 314--319. IEEE, 2012.

\bibitem{roy2010optimization}
Sandip Roy, Bhargab~B Bhattacharya, and Krishnendu Chakrabarty.
\newblock Optimization of dilution and mixing of biochemical samples using
  digital microfluidic biochips.
\newblock {\em IEEE Transactions on Computer-Aided Design of Integrated
  Circuits and Systems}, 29(11):1696--1708, 2010.

\bibitem{srinivasan2004droplet}
Vijay Srinivasan, Vamsee~K Pamula, and Richard~B Fair.
\newblock Droplet-based microfluidic lab-on-a-chip for glucose detection.
\newblock {\em Analytica Chimica Acta}, 507(1):145--150, 2004.

\bibitem{srinivasan2004integrated}
Vijay Srinivasan, Vamsee~K Pamula, and Richard~B Fair.
\newblock An integrated digital microfluidic lab-on-a-chip for clinical
  diagnostics on human physiological fluids.
\newblock {\em Lab on a Chip}, 4(4):310--315, 2004.

\bibitem{thies2008abstraction}
William Thies, John~Paul Urbanski, Todd Thorsen, and Saman Amarasinghe.
\newblock Abstraction layers for scalable microfluidic biocomputing.
\newblock {\em Natural Computing}, 7(2):255--275, 2008.

\bibitem{xu2010defect}
Tao Xu, Krishnendu Chakrabarty, and Vamsee~K Pamula.
\newblock Defect-tolerant design and optimization of a digital microfluidic
  biochip for protein crystallization.
\newblock {\em IEEE Transactions on Computer-Aided Design of Integrated
  Circuits and Systems}, 29(4):552--565, 2010.

\bibitem{xu2008automated}
Tao Xu, Vamsee~K Pamula, and Krishnendu Chakrabarty.
\newblock Automated, accurate, and inexpensive solution-preparation on a
  digital microfluidic biochip.
\newblock In {\em Biomedical Circuits and Systems Conference (BioCAS)}, pages
  301--304. IEEE, 2008.

\end{thebibliography}


\appendix

\section{Counter-Example for the Algorithm of Dinh~{\etal}~\cite{dinh2014network}}
\label{sec: counter-example for algorihtm of Dinh}

In Section~\ref{sec: introduction} we gave an outline of the algorithm
proposed by Dinh~\etal~\cite{dinh2014network} for computing mixing graphs
with minimum waste. The algorithm relies on the assumption that an
optimum mixing graph (that is, the graph that minimizes waste)
for a target set with maximum precision $d$ has depth at most $d$.
This seems indeed intuitive --- yet we show that this assumption is not valid,
by constructing a target set $T$ with maximum precision $d$ that requires
depth at least $2d-1$ to be produced by a mixing graph (without waste).

Our target set is
$T = \braced{ \, 2^{-d} \,,\, ((d-1)2^d+1): (1-2^{-d})\, }$,
that is, $T$ has one droplet with concentration $2^{-d}$ and $(d-1)2^d+1$ droplets with
concentration $1-2^{-d}$. 

We first observe that there is a mixing graph of depth $2d-1$ that produces $T$ without waste.
(See Figure~\ref{fig: dinh counterexample}.)
This graph first mixes $0$ and $1$, producing two droplets $\half$.
One droplet $\half$ is mixed $d-1$ times with $0$. This creates a path with
$d-1$ mixers that produce
droplets $2^{-2}, 2^{-3},...,2^{-d}$, plus one more droplet $2^{-d}$ that is sent to the output.
The other droplet $\half$, symmetrically, is mixed repeatedly with $1$,
producing
droplets $1-2^{-2}, 1-2^{-3},...,1-2^{-d}$, plus one more droplet $1-2^{-d}$ that is sent to the output.
For $a = 2,...,d$, we can then mix each droplet 
$2^{-a}$ with $1-2^{-a}$, which produces $2(d-1)$ droplets $\half$.
Each droplet $\half$ can be then mixed repeatedly ($d-1$ times) with $1$'s,
creating a tree-like subgraph of depth $d-1$ that produces
$2^{d-1}$ droplets $1-2^{-d}$. Thus, overall, we produce
one droplet $2^{-d}$ and $(d-1)2^d+1$ droplets $1-2^{-d}$, which is exactly $T$.
The depth of this graph is $2d-1$.

	
\begin{figure}[ht]
	\begin{center}
		\includegraphics[width = 5.2in]{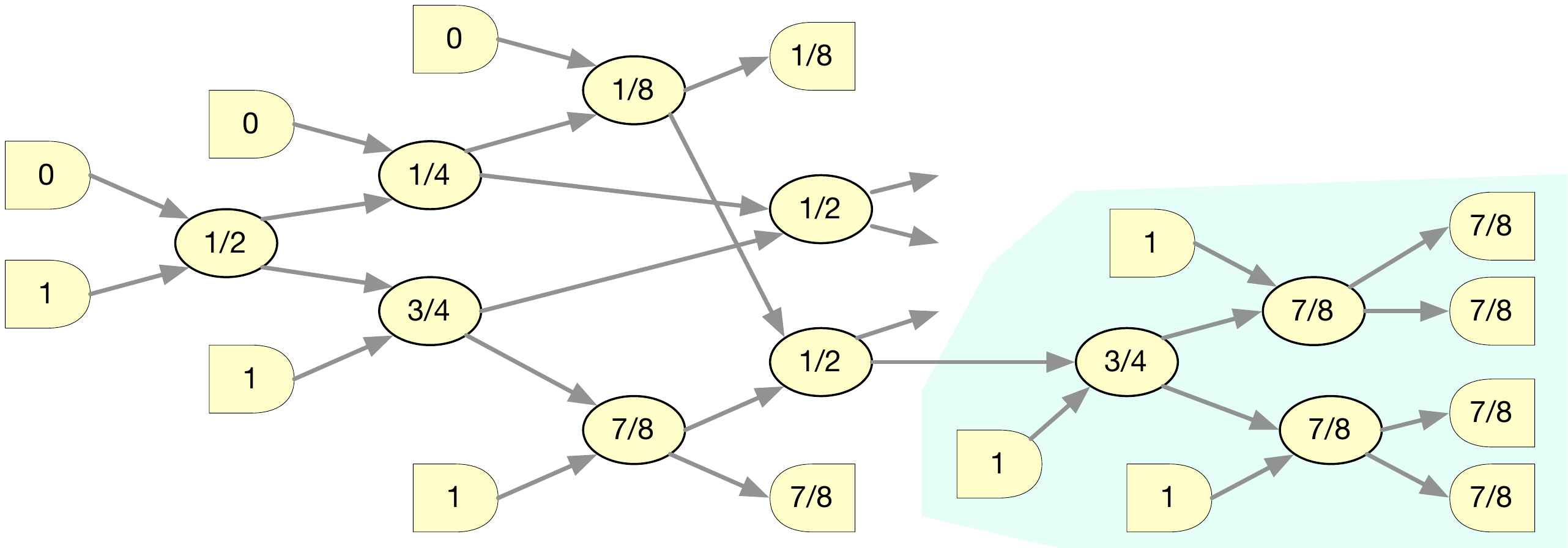}
		\caption{A mixing graph of depth $d=5$ that
		 	produces $T = \braced{ \, \oneeighth \, ,\, 17 : \seveneighths \,}$.
			Only one subgraph converting $\half$ into droplets $\seveneighths$ is shown (on shaded background).}
		\label{fig: dinh counterexample}
	\end{center}
\end{figure}

Next, we claim that producing $T$ without any waste requires depth $2d-1$.
The argument is as follows. Consider a micro-mixer that produces 
droplet $2^{-d}$ of $T$. This mixer actually produces \emph{two} droplets $2^{-d}$
and is at depth at least $d$, because it takes $d$ steps to dilute $1$ to
concentration $2^{-d}$. The fluid in the second droplet $2^{-d}$ (the one not in $T$) 
produced by this mixer
must also end up in some droplets of $T$ and its concentration will increase
to $1-2^{-d}$. The \emph{buffer concentration} in this droplet is $1-2^{-d}$ and it
takes at least $d-1$ steps to reduce it to $2^{-d}$, in order to
obtain a droplet with reactant concentration $1-2^{-d}$. We can therefore
conclude that the depth of any mixing graph for $T$ is at least $2d-1$.


\section{$\NP$-hardness proof of \MixReachability~for constant-depth mixing graphs}
\label{sec: NP-hardness}



In this section we show that the variant of {\MixReachability}
where a mixing graph of constant depth is sought is $\NP$-hard.
First we prove that \MixReachability~is $\NP$-hard for mixing graphs
of depth at most $1$, and later we explain how to extend it to
graphs of arbitrary constant depth.
	
Formally, let the \textsc{Depth-1-\MixReachability} problem be the following: 
Given two configurations $I$ and $T$, determine whether $T$ is reachable from $I$ 
via a mixing graph of depth at most one. (The depth is defined as the maximum
number of nodes on a path from an input to an output. So a mixing graph of depth $1$
does not have any edges between mixers --- each mixer is connected to two input
nodes and two output nodes.)

Recall that \textsc{Numerical-3D-Matching} is defined as follows: 
Given three multisets $X, Y, Z$ of non-negative integers such that $|X| = |Y| = |Z| = m$, 
and a non-negative integer $S$, determine whether $(X,Y,Z)$ 
has a 3D-matching consisting of triples each adding up to $S$.
(A 3D-matching of $(X,Y,Z)$ is defined as a partition $M$ of $X\cup Y\cup Z$
into $m$ triplets of the form $(x,y,z)\in X\times Y \times Z$.)
\textsc{Numerical-3D-Matching} is well-known to be $\NP$-complete.


\begin{theorem}
	\label{thm: np-hardness depth 1}
	The \textsc{Depth-1-\MixReachability} problem is $\NP$-hard.
\end{theorem}

\begin{proof}
We prove the theorem by giving a polynomial-time reduction from \textsc{Numerical-3D-Matching}.
Let
$X = \braced{x_i}_i$, $Y = \braced{y_i}_i$, and $Z = \braced{z_i}_i$ be the sets from an
instance of \textsc{Numerical-3D-Matching}, as defined above.
We construct two configurations $I$ and $T$ as follows.
For each $i = 1, 2, \dots, m$:
	\begin{itemize}
		\item $I$ contains one droplet $a_i = 2 x_i + \frac{1}{2}$ and one droplet $b_i = 2  y_i + 1$.
		\item $T$ contains two droplets with concentration $c_i = S - z_i + \frac{3}{4}$.
	\end{itemize}
	
We claim that there exists a 3D-Matching $M$ of $(X,Y,Z)$ consisting of triples that
add up to $S$ if and only if $T$ is reachable from $I$ via a mixing graph of depth at most one.


\smallskip
\noindent
 ($\Longrightarrow$)	
	Assume $M$ is a 3D-Matching of $(X, Y, Z)$ where each triplet $(x_i, y_j, z_k)\in M$ adds up to $S$.
	For each $(x_i,y_j,z_k)\in M$, we have 
	\begin{equation*}
		\textstyle
			\half (a_i + b_j) \; = \; \half [\, (2x_i+\frac{1}{2}) + (2y_j + 1)\, ]
			\; = \; x_i + y_j + \frac{3}{4}
			\; = \; S - z_k + \frac{3}{4}
			\; = \; c_k.
	\end{equation*}
	Create a mixing graph $G$ where
	for each $(x_i, y_j, z_k)\in M$ we create a mixer node with inputs $a_i$ and $b_j$
	and two outputs $c_k$. Then $G$ converts $I$ into $T$.

\smallskip
\noindent
($\Longleftarrow$)	
	Assume that there is a mixing graph $G$ of depth at most $1$ that converts $I$ into $T$.
	All numbers in $I$ are either half-integral $a_i$'s or integral $b_j$'s,
	and all numbers $c_k$ in $T$ have fractional part $\frac{3}{4}$. So $G$ must consist of
	$m$ mixer nodes, where each node has incoming edges from some $a_i \in I$ and some $b_j\in I$ and outgoing
	edges to two $c_k$'s in $T$, where $c_k = \half (a_i + b_j)$.	
	Create a 3D-Matching $M$ as follows: for each such node include 
	the corresponding triple $(x_i, y_j, z_k)$ in $M$.
	By simple calculation (reversing the calculation in implication ($\Longrightarrow$)), 
	we get that $x_i + y_j + z_k = S$.
Thus $M$ is indeed a correct solution to the instance of \textsc{Numerical-3D-Matching}.
\end{proof}

The proof in Theorem~\ref{thm: np-hardness depth 1} 
can be extended to mixing graphs with any constant depth $\sigma\ge 2$.
The idea is to modify the reduction from the theorem above, by
including in $I$ droplets $a_i = 2^\sigma x_i + 2^{-\sigma}$ and $b_i = 2^\sigma y_i + 1$,
for $i=1,...,m$, plus $m(2^\sigma-2)$ droplets with concentration $0$.
In $T$ we include $2^\sigma$ droplets of concentration 
$c_i = S - z_i + 2^{-\delta} + 2^{-2\delta}$, for $i=1,...,m$. 
The correctness proof of this reduction is left to the reader.

\end{document}